\documentclass[a4paper, 12pt, reqno]{amsart}
\usepackage{a4wide}
\usepackage[utf8]{inputenc}
\usepackage{mathptmx}       
\usepackage{helvet}         
\usepackage{courier}        
\usepackage{type1cm}

\usepackage{amsmath,amsfonts,amsbsy,amsthm}
\usepackage{eucal}
\usepackage{dsfont}
\usepackage{comment}

\usepackage{graphicx}        
                             
\usepackage{multicol}        
\usepackage[bottom]{footmisc}
\usepackage[normalem]{ulem}	
\usepackage{hyperref}  
\usepackage{soul}

\newcommand{\F}{P}

\theoremstyle{plain}
\newtheorem{theorem}{Theorem}[section]
\newtheorem{proposition}[theorem]{Proposition}

\newtheorem{lemma}[theorem]{Lemma}

\theoremstyle{remark}
\newtheorem{remark}[theorem]{Remark}

\newcommand{\e}{{\mathrm e}}
\newcommand{\1}{{\bf 1}}

\newcommand{\R}{{\mathbb R}}

\newcommand{\x}{{\bf x}}

\newcommand{\y}{{\bf y}}
\newcommand{\C}{{\mathbb C}}

\numberwithin{equation}{section}

\newcommand{\N}{{\mathbb N}}
\newcommand{\Z}{{\mathbb Z}}

\newcommand{\iu}{\mathrm{i}}
\newcommand{\di}{\mathrm{d}}

\DeclareMathOperator{\Tr}{Tr}

\begin{document}

\title[On the magnetic perturbation theory for Chern insulators]{On the magnetic perturbation theory \\ for Chern insulators}

\author{Horia D. Cornean and Massimo Moscolari}

\begin{abstract}
The gauge covariant magnetic perturbation theory is tailored for one-body \linebreak Schr\"odinger operators perturbed by long-range magnetic fields. In this work we present a self-contained exposition of the method, by outlining its technical foundations and discussing the physical heuristics behind the proofs. We apply it in order to prove the stability of spectral gaps and to study the location of the discrete spectrum.  We also analyze the (lack of) continuity with respect to the magnetic field of spectral projections corresponding to finite spectral islands, which is a particularly important situation for systems modelling Chern insulators. Finally, we show how to construct approximate projections that have an explicit dependence with respect to the magnetic field parameter. 
\end{abstract}
\maketitle

\section{Introduction}
	{Magnetic Schr\"odinger operators} are ubiquitous in modern mathematical physics. In view of their physical relevance, they gained a mathematical reputation on their own as a distinguished playground to develop and test new techniques in diverse area of mathematics, from operator and perturbation theory to non-commutative geometry. In this work we show how to deal with long-range magnetic perturbations using the so-called gauge covariant magnetic perturbation theory.

\subsection{Long-range magnetic perturbations}

We begin by describing two simple situations that highlight the singular character of magnetic perturbations. In dimension $d=2$, the free Laplacian $-\Delta$ is a selfadjoint operator on the Sobolev space $H^2(\R^2)$ and its spectrum is purely absolutely continuous and coincides with $[0,\infty)$. When we allow a non-zero constant {magnetic field} oriented perpendicular to the plane, we have to consider the operator $(-\iu \nabla -b A_L)^2$, $b>0$, where \mbox{$A_L:\R^2 \to \R^2$} is the {magnetic vector potential} in the so-called symmetric gauge $A_L(\x)=\frac{1}{2}(-x_2,x_1)$. This is the well-known {Landau Hamiltonian}, which is selfadjoint on the magnetic Sobolev space $H^2_{A_L}(\R^2)$, given by
\begin{equation}
\label{eq:H2A}
H^2_{A_L}(\R^2)=\overline{\left\{ f \in C^{\infty}_0(\R^2) \,  | \,  (-\iu \nabla - bA_L)^2 f \in L^2(\R^2) \right\}}^{\|\cdot \|_{D_{A_L}}}
\end{equation}
where $C^{\infty}_0(\R^2)$ denotes the set of smooth functions with compact support and the closure is taken with respect to the graph norm $\|f\|^2_{D_{A_L}}:=\|f\|^2+ \|(-\iu \nabla - bA_L)^2 f\|^2$.  Its spectrum is purely pure point and given by
$$
\big\{\left(2n+1\right)\, b,\quad  n \in \N \big\}
$$
where each eigenvalue is infinitely degenerate. The analysis of the Landau Hamiltonian \cite{Fock,Landau} is by now textbook material \cite{LandauLifshitz}, see \cite{AvronHerbstSimon} or \cite{MoscolariPanati2019} for a  recent review on the subject. This simple example shows that adding a small magnetic field to the system completely changes both the domains of self-adjointness and the nature of the spectrum. The reader familiar with classical physics will not find this surprising: a charged particle confined to a plane and subjected to a constant magnetic field perpendicular to the plane moves on confined circular orbits. The quantum analogue of such orbits are exactly the bound states corresponding to the Landau levels. From a mathematical point of view, the culprit behind this behaviour is the linear growth at infinity of the magnetic potential. If we formally write the difference $ (-\iu \nabla -bA_L)^2- (-\Delta)$, we get the first order differential operator
$$
W_b:=\iu 2b A_L \cdot  \nabla  + b^2 A_L^2 
$$
whose coefficients grow polynomially at infinity. This strongly suggests  that the ``magnetic perturbation" $W_b$ is neither form, nor operator bounded  with respect to the free Laplacian. In general it can be proved that the perturbed magnetic Hamiltonian fails to converge in the norm resolvent sense to the unperturbed non-magnetic Hamiltonian when the constant magnetic field goes to zero, see \cite[Theorem 6.3]{AvronHerbstSimon}. 

It turns out that most of the singular behaviour induced by long-range magnetic fields is contained in some highly oscillating unimodular phase factors. In the early 2000s, G. Nenciu and the first Author developed a general method which deals with these singular phases in a systematic way \cite{CorneanNenciu1998, Nenciu2002}. This framework is nowadays called {gauge covariant magnetic perturbation theory}. The main purpose of this paper is to outline the main ideas behind this method, together with some physically and mathematically relevant applications.

\subsection{Mathematical framework and results}
In view of the physical applications to two-\linebreak dimensional condensed matter systems, we focus our attention on the Hilbert space $L^2(\R^2)$. However, all the results presented in this work can be extended to the case of $L^2(\R^3)$. 

Let $b \in \R$ and consider the family of operators
$$
H_b:=(-\iu \nabla - b A+\mathcal{A})^2 +V
$$
where $\nabla=(\partial_1,\partial_2)$ denotes the usual gradient operator, $V: \R^2 \to \R$ is a scalar potential, $A: \R^2 \to \R^2$ is a  magnetic vector potential corresponding to a smooth and bounded magnetic {\it field}, not necessarily constant, and $\mathcal{A}:\R^2 \to \R^2$ is a smooth and bounded magnetic vector {\it potential}, which in the periodic case would model a background magnetic field with zero-mean flux. We assume that $V$ is uniformly locally in $L^{2+\delta}(\R^2)$, namely $V \in L^{2+\delta}_{\textrm{u.loc}}(\R^2)$ for some $\delta>0$ (see \eqref{eq:uloc} for the precise definition), while $A,\mathcal{A}\in C^{\infty}(\R^2,\R^2)$. The parameter $b \in \R$ controls the strength of the magnetic perturbation. For the sake of the presentation, in the following we will only consider  the case $b>0$. Notice that the total magnetic field of the system is given by $B(\x) + b(\x)$, where $B(\x)=\nabla \times \mathcal{A}(\x)$ is the background magnetic field and  $b(\x)=b \, \nabla \times A(\x)$ is the magnetic field perturbation. Furthermore, we assume that 
\begin{equation}
\label{eq:HypB}
\sup_{\x\in\R^2}|\partial^{\alpha} b(\x)|\leq b\,  b_\alpha,\quad \alpha\in \mathbb{N}^2,\quad |\alpha|\leq 1,\quad 0\leq b_\alpha<\infty,\quad 0\leq b\leq 1.
\end{equation}

Since $V$ is infinitesimally bounded with respect to the free Laplacian \cite[Theorem XIII.96]{ReedSimonIV} and the magnetic potentials are smooth, by using the diamagnetic inequality \cite[Theorem 2.3]{AvronHerbstSimon} and the standard theory of Schr\"odinger operators we get that $H_b$ is essentially selfadjoint on $C^\infty_0(\R^2)$ and its domain of selfadjointness is the magnetic Sobolev space $H^2_A(\R^2)$, see \eqref{eq:H2A} for the definition of the second magnetic Sobolev space.

We will consider an unperturbed Hamiltonian $H_{b_0}$, corresponding to $b=b_0$ where $b_0$ might be zero,  and we are interested in what happens when $b$ varies around $b_0$. Standard questions in perturbation theory concern the location and behaviour of the spectrum of the operator $H_b$ as $b$ varies. In particular:
\begin{enumerate}
    \item \label{StabilityGaps} Let $I \subset \R$  be a closed set contained in the resolvent set $\rho(H_{b_0})$. Is $I$ still contained in the resolvent set of $\rho(H_b)$ for a suitably small $|b-b_0|$? This question is usually rephrased as: are the {\emph{spectral gaps}} of the operator $H_{b_0}$ stable with respect to long-range magnetic perturbations?
    \item \label{LocationEigenvalue} Let $e_0$ be an isolated eigenvalue of $H_{b_0}$ of finite multiplicity. What is the fate of $e_0$ as $b$ changes?
\end{enumerate}

These questions have been thoroughly addressed in the literature by using several technical tools and under various regularity assumptions on the scalar and magnetic potentials:
\begin{itemize}
    \item Question \ref{StabilityGaps} has been first answered by Avron and Simon \cite{AvronSimon} for the case of Schr\"odinger operators with periodic scalar potentials perturbed by a constant magnetic field, and then in full generality by Nenciu \cite{Nenciu1986}, see also \cite{Nenciu1991}.  \item Question \ref{LocationEigenvalue} can also be tackled in several ways. Already in the standard textbook  by Kato \cite{Kato}, it is shown that  isolated eigenvalues of finite multiplicity of $H_{b_0}$ can either stay as an isolated eigenvalue for $H_b$, or split in different eigenvalues according to the multiplicity of the original one. Detailed asymptotic expansions have been proved by Avron, Herbst and Simon, see for example \cite{AvronHerbstSimon}, also in the setting of the Zeeman effect, namely considering the perturbation of bound states of hydrogen-like atoms. Moreover, a convergent ``power series"  expansion in $(b-b_0)$ where the coefficients still depend on $b$ has been given by Briet and the first Author \cite{BrietCornean2002} and more recently by Savoie \cite{Savoie}, see also \cite{Nenciu2002} for the case of a simple isolated eigenvalue.  

\end{itemize}

Besides the {location of the spectrum} of the magnetic operator, there are other important questions that have been thoroughly analyzed, such as the spatial decay of the eigenfunctions \cite{CorneanNenciu1998}, the regularity of the integral kernels of various functions of the magnetic operators \cite{Simon1982,CorneanNenciu2009}, the nature of the spectrum and the selfadjointness of the magnetic operator with singular potentials just to mention a few \cite{BrietCornean2002}. A successful strategy to solve these types of problems has been the gauge covariant magnetic perturbation theory: In the following, we will briefly review the fundamental aspects of this method and, as a demonstrative example, we show how to answer Questions~\ref{StabilityGaps} and \ref{LocationEigenvalue} in a unified framework.

In addition to its mathematical importance, the stability of spectral gaps plays a crucial role in the modeling of topological insulators, particularly in the modeling of the so-called {Chern insulators}. The mathematics of topological insulators is one of the most active research topics in the mathematics of quantum mechanics. Loosely speaking, topological insulators are peculiar materials that are insulating in their bulk, namely their bulk model has a spectral gap, but support edge states with a quantized conductance. Furthermore, the archetypal example of topological insulators is given by materials that exhibit the integer quantum Hall effect, whose main mathematical model is provided by the Landau Hamiltonian. 
Notice that the spectral gap of the bulk model implies only the vanishing of the direct conductivity. Indeed, the transverse conductivity can be non-zero and its values must be quantized. The connection between the quantized bulk transverse conductivity and the quantized conductance of the edge states is at the core of the bulk-edge correspondence \cite{SchulzBaldesKellendonkRichter,KellendonkSchulzBaldes,ElbauGraf,ElgartGrafSchenker}, see also the  recent work \cite{CorneanMoscolariTeufel} for a proof of bulk-edge correspondence using gauge covariant magnetic perturbation theory and a more detailed discussion of the literature.

Building upon the pioneering works of Avron, Seiler, Simon \cite{AvronSeilerSimon1983,AvronSeilerSimon} and Bellissard, van Elst, and Schulz-Baldes \cite{Bellissard1986, BellissardVanElstSchulzBaldes}, in the last decades it has been shown that the bulk transport properties of topological insulators are encoded in their Fermi projection (see for example \cite{MarcelliMoscolariPanati} and the references therein), that is, the {spectral projection} onto the occupied states of the system. In Chern insulators with sufficiently weak disorder, the Fermi energy lies in a spectral gap and the projection onto the occupied states of the system is exactly the spectral projection onto the states below the spectral gap. In this framework, the key quantity is the Chern number (or its non-periodic generalization, the Chern character) of the Fermi projection: Let $P$ be a spectral projection of $H_b$ with an integral kernel $P(\x;\y)$ that is exponentially localized and $X_i$, $i \in \{1,2\}$ be the position operators. Then one can show (\cite{MarcelliMoscolariPanati}, see also \cite{AvronSeilerSimon}) that the double commutator $\left(P[[P,X_1],[P,X_2]]\right)$ is a locally trace class operator with an exponentially localized integral kernel such that the limit
\begin{equation}
\label{eq:Chern}
C(P)=\lim_{L\to \infty} \frac{2\pi \iu}{4L^2} \int_{[-L,L]^2} \left(P[[P,X_1],[P,X_2]]\right)(\x;\x) \,\mathrm{d}\x 
\end{equation}
exists and it is an integer. The quantity $C(P)$ is called the Chern character of the projection $P$, which can also be interpreted as the trace per unit surface of the double commutator $\iu \left(P[[P,X_1],[P,X_2]]\right)$. The name  ``Chern" originates from the fact that \eqref{eq:Chern} coincides with the Chern number of the Bloch bundle generated by $P$, whenever $H_b$ commutes with a unitary representation of the group $\Z^2$ provided by magnetic translations \cite{Panati}. If $P$ is the Fermi projection with Fermi energy in a spectral gap, its integral kernel is exponentially localized, see \eqref{eq:PExpo}, and $C(P)$ is connected to the quantum transport coefficients of the physical system modelled by $H_b$. In particular, the Hall conductivity of the system computed in the linear response regime is proportional to $C(P)$ (see for example the recent papers \cite{DeRoeckElgartFraas,MarcelliMonaco} and references therein). For this reason, these materials are also called Chern insulators. Since the integer quantum Hall effect is mainly a two-dimensional phenomenon (even though other types of topological insulators exist in all dimensions), we restrict our analysis to magnetic Hamiltonians in $L^2(\R^2)$.

In Section \ref{sec:StabilityGaps} we address Question \ref{StabilityGaps} by showing the stability of spectral gaps under small general magnetic field perturbations. Such stability is a consequence of the identity
\begin{align*}
(H_b-z {\bf 1})^{-1}=S_\epsilon(z)\big (\1+T_\epsilon(z)\big )^{-1}.
\end{align*}
where $S_\epsilon(z)$ and $T_\epsilon(z)$ are bounded operators with an explicit dependence on the magnetic field parameter $b$ ($\epsilon$ is defined to be $b-b_0$), see also \eqref{eqn:T(z)}. 

In Section \ref{sec:Projections}, we exploit gauge covariant magnetic perturbation theory in order to control the magnetic dependence of the Fermi projection, by describing a general perturbation scheme which was recently successfully applied by the Authors and their collaborators to the analysis of the transport and topological properties of topological insulators \cite{CorneanMonacoMoscolari2019,MoscolariMonaco,CorneanMonacoMoscolari2021, CorneanMoscolariTeufel}. 

As a by-product, in Section \ref{sec4} we show that the spectral projections of magnetic Schr\"odinger operators are always continuous with respect to the magnetic field in the strong operator topology, while in Section \ref{sec5} we show that they fail to be continuous in the uniform topology at least when they are topologically non-trivial, that is whenever their Chern character is not zero. An explicit example for such a failure is given by the projection on the ground state of the purely magnetic Landau Hamiltonian. 

In Section \ref{sec:Eigenvalue} we focus our attention on Question \ref{LocationEigenvalue}, and we show that magnetic spectral projections onto isolated eigenvalues are always smooth with respect to the magnetic field, even in the trace norm topology. Furthermore, for simple eigenvalues  we show that they obey a fixed point equation \eqref{hcj10} induced by a Feshbach formula, whose iterations produce an asymptotic expansion in $b-b_0$ in a very efficient way.

\medskip

Before going into the details of the theory, we single out two properties which will help us in developing the method:
\begin{enumerate}
    \item The resolvent $(H_b-z\1 )^{-1}$ is an integral operator with an integral kernel that is jointly continuous outside of the diagonal and such that 
    for every $K_b$ compact subset of the resolvent set of $H_b$ there exist $C,\alpha >0$ independent of $0\leq b\leq 1$ such that for all $\x\neq \x' \in \R^2$:
    \begin{equation}
\label{IntKernelResolvent}
\sup_{z \in K_b}\left| (H_b-z {\bf 1})^{-1}(\x;\x') \right| \leq C \ln\left (2+\|\x-\x'\|^{-1}\right )e^{-\alpha\|\x-\x'\|} \, .
\end{equation}
    \item Moreover, for all $\x\neq \x' \in \R^2$:
    \begin{equation}
\label{eq:estimateResolvent}
\sup_{z\in K_b}\left| (-\iu \nabla_\x -bA(\x))(H_b-z {{\bf 1}})^{-1}(\x;\x') \right| \leq C  \|\x-\x'\|^{-1}e^{-\alpha\|\x-\x'\|} \,.
\end{equation}
 \end{enumerate}

In the rest of the paper we take the properties \eqref{IntKernelResolvent} and \eqref{eq:estimateResolvent} for granted, and in \mbox{Appendix~\ref{sec:Estimates}} we will briefly explain how to prove such statements  and decay estimates for the magnetic Schr\"odinger operators we are considering. We stress that these assumptions are fairly general and rely on the the fact that we are dealing with elliptic partial differential operators. Indeed, a suitable magnetic perturbation theory has also been developed for Dirac operators and operators defined on domains with boundary, see \cite{CorneanMoscolariTeufel,CorneanMoscolariSorensen}.

\section{The quasi-inverse and stability of spectral gaps}
\label{sec:StabilityGaps}
In this section we review the basics of gauge covariant magnetic perturbation theory, the presentation is based on \cite{CorneanNenciu1998,CorneanNenciu2000, Nenciu2002, CorneanMonacoMoscolari2021}. Such a perturbation theory is inspired by the well-known Peierls substitution and it relies on the gauge invariance of the electromagnetic theory. A heuristic discussion of the theory is postponed to Section \ref{sec:Heuristic}.

Since the addition of $\mathcal{A}$ generates a perturbation that is relatively bounded with respect to the free Laplacian, without loss of generality we may assume that the bounded magnetic vector potential $\mathcal{A}(\x)=0$ while the magnetic field is purely long-range and of the form  $$b(\x)=b\,  \mathfrak{b}(\x),\quad 0\leq b\leq 1,$$  
where $\mathfrak{b}$ obeys \eqref{eq:HypB} with $b=1$. Notice that $b\in \R$ and $\mathfrak{b}$ is real valued. Let us start by constructing a special family of magnetic vector potentials $\left\{{A}_\y: \R^2 \to \R^2\right\}_{\y \in \R^2}$ using the two-dimensional vector notation $\x=(x_1,x_2) $ 
\begin{equation}
\label{eq:TGauge}
{A}_\y (\x):= \left( \int_0^1 s  \, \mathfrak{b}(\y+s(\x-\y))\mathrm{d}s \right) \; (-x_2+y_2,x_1-y_1).
\end{equation} 
In the following we will use the shorthand notation ${A}_\mathbf{0}=A$.
Notice that, if we consider our system embedded in $\R^3$, then the magnetic field is given by the vector field $\vec{b}(\x)=(0,0,b(\x))$ and $A_\mathbf{0}$ is, with a little abuse of notation, nothing but the usual transversal gauge 
\begin{equation*}
A_\mathbf{0}(\vec{x})=  \left( \int_0^1 s \mathfrak{b}(s\vec{x})\mathrm{d}s \right) \wedge \vec{x} .
\end{equation*}

By using \eqref{eq:HypB} and \eqref{eq:TGauge}, we get the estimate
\begin{align*}
\|\partial_\x^\alpha {A}_\y (\x) \|\leq b_\alpha \;\|\x-\y\|,\quad \alpha\in \mathbb{N}^2,\quad |\alpha|\leq 1. 
\end{align*}

We have the following simple but fundamental lemma. 
\begin{lemma}
\label{lemma:MPhase}
The family of magnetic vector potentials $\left\{A_\y: \R^2 \to \R^2\right\}_{\y \in \R^2}$ generates the same magnetic field $\mathfrak{b}(\x)$
\begin{equation*}
\nabla \times {A}_\y(\x)=\nabla \times {A}(\x)= \mathfrak{b}(\x) \, .
\end{equation*}
As a consequence, the magnetic potentials $\left\{{A}_\y\right\}_{\y \in \R^2}$ differ by a pure gauge factor, namely there exists $\varphi_\y(\x) \equiv \varphi(\x,\y)$ such that
\begin{equation}\label{hc1}
{A}(\x) - {A}_\y(\x) =\nabla_\x \varphi(\x,\y)
\end{equation}
and $\varphi$ is given by the line integral of the vector potential ${A}$ along the straight-line segment connecting $\y$ and $\x$
\begin{equation}
\label{eq:PeierlsPhaseAux}
\varphi(\x,\y):=\int_{\y}^{\x}{A} \cdot \mathrm{d}\mathbf{s} \equiv \int_{0}^{1}  \, {A}(\y+s(\x-\y)) \cdot (\x-\y) \,\mathrm{d}s .
\end{equation}
Furthermore, we have the expression
\begin{equation}
\label{eq:PhaseAux}
\varphi(\x,\y)=-\left(\int_{0}^{1} \int_{0}^{1} \, t\mathfrak{b}(t(\y+s(\x-\y))) \mathrm{d}t \mathrm{d}s \right) (x_1y_2-x_2y_1) \, .
\end{equation}
\end{lemma}
\begin{proof}
By explicit computation using the chain rule for the derivatives, one can check that the magnetic potentials ${A}_\y (\x)$ and ${A}(\x)$ generate the same magnetic field.

Then, define 
$$
\mathbf{f}(\x,\y):= {A}(\x) - {A}_\y(\x) .
$$
We have that there exists $\varphi_\y(\x)$ such that $\mathbf{f}(\x,\y)=\nabla_\x \varphi_\y(\x)$. Clearly, $\varphi_\y(\x)$ is not unique. If we impose that $\varphi_\x(\x)=0$, it is enough to define
$$
\varphi_\y(\x)=\varphi \left(\x, \y\right)=\int_{y_1}^{x_1} f_1\left(t, x_2 ; \y\right) \mathrm{d} t +\int_{y_2}^{x_2}  f_2\left(y_1, t ; \y\right) \mathrm{d}t \, .
$$
Indeed, consider the straight-line segment 
$$
\gamma\left(\x, \y\right):=\left\{\x'(t)=\y+t\left(\x-\y\right) \, | \, t \in [0,1]\right\}.
$$
The line integral of $\mathbf{f}$ along $\gamma\left(\x, \y\right)$ is given by 
$$
\int_{\gamma\left(\x, \y\right)} \mathbf{f}(\mathbf{s};\y) \cdot \mathrm{d}\mathbf{s} = \int_{\gamma\left(\x, \y\right)} \left({A}(\mathbf{s}) - {A}_\y(\mathbf{s})\right) \cdot \mathrm{d}\mathbf{s} = \int_{\gamma\left(\x, \y\right)} {A}(\mathbf{s})  \cdot \mathrm{d}\mathbf{s}
$$
where in the last equality we have used the fact that $A_\y(\x) \cdot (\x-\y) =0$ in view of the transversality of the gauge, see \eqref{eq:TGauge}. Thus, since $\mathbf{f}(\x,\y)=\nabla_\x \varphi_\y(\x)$, we get
\begin{equation}
\label{eq:Aux12}
\varphi(\x,\y)=\int_{\gamma\left(\x, \y\right)} {A}(\mathbf{s})  \cdot \mathrm{d}\mathbf{s} =\int_{\y}^{\x}{A} \cdot \mathrm{d}\mathbf{s}\, .
\end{equation}
The expression \eqref{eq:PhaseAux} follows by inserting \eqref{eq:TGauge} in \eqref{eq:Aux12}. 
\end{proof}

The phase $\varphi: \R^2 \times \R^2 \to \R$ defined in Lemma \ref{lemma:MPhase} is usually called Peierls phase. Important properties of the Peierls phase are recalled in the following proposition.

\begin{proposition}
The function $\varphi :\R^2 \times \R^2 \to \R$ satisfies the following properties:
\begin{enumerate}
    \item $\varphi(\x,\y)=-\varphi(\y,\x)$.
    \item $|\varphi(\x,\y)| \leq C\, \|\x\| \|\y\|$
    \item 
    Let $\langle \x,\y,\,\x' \rangle $ be the oriented triangle with vertices in $\x,\y$ and $\x'$ and define the flux of $\mathfrak{b}(\x)$ through this triangle by  $\mathrm{fl}(\x,\y,\x')=\int_{\langle \x,\y,\,\x' \rangle }  \;  \mathfrak{b}(\tilde{\x}) \cdot  \di \sigma( \tilde{\x})$. Then there exists a constant $C<\infty$ such that
    \begin{equation}
    \label{eq:CompositionId}
    \begin{aligned}
    &\varphi(\x,\y)+\varphi(\y,\x')-\varphi(\x,\x') = \mathrm{fl}(\x,\y,\x') \, ,\\
    &|\mathrm{fl}(\x,\y,\x')| \leq C\, \|\x-\y\| \|\y-\x'\|.
    \end{aligned}
    \end{equation}
\end{enumerate}
\end{proposition}
\begin{proof}
    The anti-symmetry follows from \eqref{eq:PeierlsPhaseAux}, and the polynomial bound from \eqref{eq:PhaseAux}. The ``composition identity" \eqref{eq:CompositionId} is a consequence of the Stokes theorem.   
    \end{proof}

 Using \eqref{hc1} we obtain the following central ``intertwining" identity:
\begin{align}\label{eq:FundGCMP}
\big (-\iu \nabla_\x -b{A}(\x)\big )e^{\iu b \varphi(\x,\y)} =e^{\iu b\varphi(\x,\y)}\big (-\iu \nabla_\x -b{A}_\y (\x)\big )  \,.
\end{align}

We are now ready to analyze the resolvent of the perturbed Hamiltonian $H_b$. We introduce the new variable $\epsilon :=b-b_0$. For every $z\in K \subset \rho(H_{b_0})$, we define the operator $S_\epsilon(z)$ given by the integral kernel
\begin{equation} \label{eqn:S(z)}
S_\epsilon(z)(\x;\x'):=e^{\iu \epsilon \varphi (\x,\x')}(H_{b_0}-z {\bf 1})^{-1}(\x;\x').
\end{equation}

\begin{lemma}
\label{lemma:S}
The operator $S_\epsilon(z)$ defined by the integral kernel \eqref{eqn:S(z)}
is a bounded operator. 
\end{lemma}
\begin{proof}
$|S_\epsilon(z)(\x;\x')|$ is pointwise dominated by the integral kernel of the resolvent of $H_{b_0}$ which due to \eqref{IntKernelResolvent} is dominated by an $L^1(\R^2)$-function of $\x-\x'$, thus the Schur test \cite[Theorem 5.2]{HalmosSunder} implies that $S_\epsilon(z)$ is a bounded operator. 
\end{proof}

The next result shows that for $\epsilon$ small enough, $S_\epsilon(z)$ behaves almost like the inverse operator of $(H_b-z{\bf{1}})$. 

\begin{proposition}\label{prophc1}
Let $z \in K \subset \rho(H_{b_0})$ and $b=b_0+\epsilon$. Then the range of $S_\epsilon(z)$ equals the domain of $H_b$ and
\begin{equation} \label{eqn:T(z)}
\big (H_{b_0+\epsilon} -z {\bf 1}\big )\, S_\epsilon(z)={\1} + T_\epsilon(z)
\end{equation}
where $T_\epsilon(z)$ is a bounded operator corresponding to the integral kernel
\begin{multline*}
T_\epsilon(z)(\x;\x'):=-2\epsilon \, e^{\iu \epsilon \varphi (\x,\x')}A_{\x'}(\x)\cdot (-\iu \nabla_\x -b_0{A}(\x))(H_{b_0}-z {\bf 1})^{-1}(\x;\x')\\
+e^{\iu \epsilon \varphi (\x,\x')}\left \{\epsilon^2|A_{\x'}(\x)|^2+\iu \epsilon \; {\rm div}_\x \, A_{\x'}(\x)\right \}(H_{b_0}-z {\bf 1})^{-1}(\x;\x').
\end{multline*}
Moreover $\sup_{z\in K}\|T_\epsilon(z)\|\leq C \, |\epsilon|$, and there exists $\epsilon_0>0$ such that for all $0\leq |\epsilon|\leq \epsilon_0$ we have $K\subset \rho(H_{b_0+\epsilon})$.
\end{proposition}
\begin{proof}
First, notice that for every test function $\psi,\phi \in C^{\infty}_0(\R^2)$ we have
$$
\langle (H_{b_0}-\overline{z} \, {\bf 1})\psi ,  (H_{b_0}-z {\bf 1})^{-1} \phi \rangle =\langle \psi, \phi \rangle \,  ,
$$
which implies that the integral kernel of $(H_{b_0}-z {\bf 1})^{-1}$ satisfies the distributional equality 
$$
\big ((-\iu\nabla_\x - b_0\, {A}(\x))^2 + V(\x)-z{\bf 1}\big )\, (H_{b_0}-z )^{-1}(\x;\y)=\delta(\x-\y) \, .
$$
By using \eqref{eq:FundGCMP} together with the identity $e^{\iu \epsilon \varphi(\x,\x)}=1$ and the estimates on the integral kernel from \eqref{IntKernelResolvent} and \eqref{eq:estimateResolvent}, we get that \eqref{eqn:T(z)} holds at least in the weak sense, namely for all test functions $\psi$ and $\phi$  we have 
\begin{equation}
\label{eq:WeakS}
\langle (H_b-\bar{z} {\bf 1}) \psi, S_\epsilon(z) \phi \rangle = \langle \psi , \big (\mathbf{1} + T_\epsilon(z) \big ) \phi \rangle.
\end{equation}
From the definition of the integral kernel of $T_\epsilon(z)$ we deduce that there exist $C,\alpha' >0$ such that
\begin{equation*} 
|T_\epsilon(z)(\x;\x')|\leq C \, |\epsilon|\,  e^{-\alpha' \|\x-\x'\|}.
\end{equation*}
The integral kernel is pointwise dominated by an $L^1$ function of $\x-\x'$, hence by a Schur test \cite[Theorem 5.2]{HalmosSunder}  we get that $\|T_\epsilon(z)\| \leq C \, |\epsilon|$.
Since $H_b$ is essentially selfadjoint on $C^\infty_0(\R^2)$ and $S_\epsilon(z)$ and $T_\epsilon(z)$ are bounded operators, we can extend \eqref{eq:WeakS} to any function $\psi \in \mathrm{D}(H_b)$ and any function $\phi \in L^2(\R^2)$. Since $H_b$ is selfadjoint, this implies that $S_\epsilon(z)$ maps $L^2(\R^2)$ into the domain of $H_b$ and \eqref{eqn:T(z)} is proved.

Now let us prove that the spectral gaps are stable. Here we may assume that $K\subset \R$. Let $\lambda\in K$. From equation \eqref{eqn:T(z)} we get that for $\epsilon$ small enough the operator $H_{b_0+\epsilon}-\lambda \1$ has a bounded right inverse, which by itself is not enough to show that $H_{b_0+\epsilon}-\lambda \1$ is invertible. However, if the strength of the magnetic field perturbation $|\epsilon|$ is small enough, we have that $\|T_\epsilon(\lambda+\iu \delta)\| \leq 1/2$ for all $\lambda\in K$ and $0\leq \delta\leq 1$. By using the Neumann series, we get that $\mathbf{1} +T_\epsilon(\lambda+\iu \delta)$ is invertible, and at least when $\delta>0$ we may write
\begin{equation}
\label{eq:ResolventTNeumann}
\big (H_b-(\lambda+\iu \delta) {\bf 1}\big )^{-1}=S_\epsilon(\lambda+\iu \delta)\big (\mathbf{1} +T_\epsilon(\lambda+\iu \delta)\big )^{-1}.
\end{equation}
The right hand side of the above identity remains bounded when $\delta$ tends to zero. This implies that \eqref{eq:ResolventTNeumann} also holds for $\delta=0$ because 
$$H_b-\lambda\1 =\Big ( \1+ \iu \delta \big (H_b-(\lambda+\iu \delta) {\bf 1}\big )^{-1}\Big )\, \big (H_b-(\lambda+\iu \delta) {\bf 1}\big )$$
is invertible if $\delta$ is small enough. 
\end{proof}

By multiplying both sides of \eqref{eqn:T(z)} from the left by $(H_b-z {\bf 1})^{-1}$ we obtain a resolvent-like identity:
\begin{align}\label{eq:ResolventID}
(H_b-z {\bf 1})^{-1}=S_\epsilon(z)-(H_b-z {\bf 1})^{-1}T_\epsilon(z).
\end{align}
At the level of integral kernels, \eqref{IntKernelResolvent} and \eqref{eq:estimateResolvent} imply the estimate
\[ \sup_{z \in K}| \{(H_b - z {\bf 1})^{-1} T_\epsilon(z)\}(\x;\y) | \le C\,  |\epsilon|\, e^{-\alpha' \|\x-\y\|}\, \]
which applied to \eqref{eq:ResolventID} gives:
\begin{equation}
\label{eq:MPTResolvent}
\sup_{z \in K} \left|(H_b-z {\bf 1})^{-1}(\x;\y)-S_\epsilon (z)(\x;\y)\right|\leq C\, |\epsilon|\,  e^{-\alpha' \|\x-\y\|} \, .
\end{equation}

This is the first step in analyzing the spectral behaviour under long-range magnetic perturbations. As we have shown in the introduction, the nature of the spectrum may drastically change, but we saw in Proposition \ref{prophc1} that the spectrum varies at least continuously as a set. Moreover, if $H_{b_0}$ has an isolated spectral island, it is possible to prove much stronger results on the location of its so-called spectral edges. 

An isolated spectral island is by definition some bounded part of the spectrum which is separated from the rest by two gaps. More precisely, let $a_1,a_2 \in \rho(H_{b_0})$, with $a_1<a_2$ and define
$$
\sigma_0:=(a_1,a_2) \cap \sigma(H_{b_0})\, . 
$$
Here $s_-(0) := \inf \sigma_0$ and $s_+(0) := \sup \sigma_0$ are its spectral  edges. If $b-b_0$ is small, we know from Proposition \ref{prophc1} that $a_1$ and $a_2$ remain in the resolvent set of $H_b$, hence it makes sense to define 
$$
\sigma_\epsilon:=(a_1,a_2) \cap \sigma(H_{b_0+\epsilon})\, ,\quad |\epsilon|\ll 1. 
$$

One may show that  the edges of spectral islands of magnetic Schr\"odinger operators are actually Lipschitz in the magnetic field: 
\begin{proposition}[\cite{Cornean2010}]
\label{Cornean2010}
Denote by $s_-(\epsilon) := \inf \sigma_\epsilon$ and $s_+(\epsilon) := \sup \sigma_\epsilon$ the edges of the spectral island. Then there exist $\epsilon_0>0$ and $c>0$ such that $|s_{\pm}(\epsilon) - s_{\pm}(0)| \leq  c |\epsilon|$ for all $|\epsilon|\leq \epsilon_0$.
\end{proposition}

The same sharp result has been recently extended in \cite{CorneanPurice} and \cite{CorneanGardeStottrupSorensen} to the setting of magnetic  pseudodifferential calculus. We note that when $V$ is smooth and bounded, then $H_b$ can be seen as the magnetic Weyl quantization of an elliptic H\"ormander symbol of type $S^2_1(\R^4)$. Nevertheless, when the coefficients are rough, working directly with integral kernels, as we do here, and not going through symbols, is much more convenient.  

\subsubsection*{Constant magnetic field perturbations}
In physical applications, one is usually interested in perturbations arising from a constant magnetic field, which means that $\mathfrak{b}(\x)$ is just the constant unit field perpendicular to the plane. In such a case, the antisymmetric Peierls phase becomes linear in both variables and simplifies to 
\begin{equation}\label{hcj1}
\varphi(\x,\y)=-\frac{1}{2}(y_2 x_1 - x_2 y_1). 
\end{equation}
Moreover, due to the linearity of the magnetic potential, we have $A_\y(\x)=A_L(\x-\y)$. Then, we obtain the identity
\begin{align*}
(-\iu \nabla_\x -b{A}_L(\x))e^{\iu b \varphi(\x,\y)}=e^{\iu b \varphi(\x,\y)}(-\iu \nabla_\x -b{A}_L(\x-\y))  \, .
\end{align*}
 Moreover, in this case the antisymmetric Peierls phase satisfies the composition rule
\begin{equation}\label{dhc1}
    \varphi(\x,\y)+\varphi(\y,\x')-\varphi(\x,\x') =\varphi(\x-\y,\y-\x') = \mathrm{fl}(\x,\y,\x')\, .
    \end{equation}

\subsection{Some heuristics behind the gauge covariant magnetic perturbation theory}
\label{sec:Heuristic}
The reader with a background in physics will have already noticed a similarity between the construction of the quasi-inverse  $S_\epsilon(z)$, and how magnetic fields are introduced in tight-binding Hamiltonians via Peierls phases. The interested reader is referred to \cite{Nenciu1991} for more details and physical aspects. Instead, here we would like to give a more mathematical and practical heuristics at the heart of the theory. The discussion is based on the review paper \cite{Nenciu2002}. 

The first observation is again coming from physics and is the fact that the electromagnetic theory is gauge invariant: as we have already seen before, given two magnetic potentials $A$ and $A'$, they generate the same magnetic field if and only if there exists $\theta: \R^2 \to \R$ such that $A-A'=\nabla \theta $. In this discussion, we will avoid any regularity issues and we assume to work with regular potentials only. The magnetic field generated by the two potentials is the same, however the Hamiltonians associated with $A$ and $A'$ are, in principle, different. Indeed, if we define the bounded operator $U_\theta$ by 
$$
U_\theta \eta (\x)= e^{\iu \theta(\x)} \eta (\x) \, \, \, \, \,  \forall \, \eta \in L^{2}(\R^2) , 
$$
we have that $U_\theta (-\iu \nabla_\x - A')^2 U_\theta^{-1}  = (-\iu \nabla_\x - A)^2$ . It can be shown that the operator $U_\theta$ actually provides a unitary equivalence between  $H^2_{A'}(\R^2)$ and  $H^2_{A}(\R^2)$, thus showing that $(-\iu \nabla_\x - A')^2$ and $(-\iu \nabla_\x - A)^2$ are unitarily equivalent. $U_\theta$ is called gauge transform and the function $\theta$ is usually called the generator of the gauge transform. Recall that the key formula \eqref{eq:FundGCMP} is nothing but the gauge covariance of the magnetic momentum $(-\iu \nabla -A)$ under the gauge transform generated by $\theta(\x)=\varphi(\x,\x')$.

Even though different gauges lead to different Hamiltonians, the fact that all these operators are unitarily equivalent guarantees that observable quantities, such as the spectrum, are gauge invariant.  This observation suggests to develop a perturbation scheme only involving gauge invariant objects only depending on the magnetic field, and this is exactly what the gauge invariant perturbation theory is.

Let us consider the Hamiltonian $H_b$ and denote it by $H_b(A)$, in order to keep track of the gauge we are using. Notice that we are only interested in the gauge freedom related to the perturbing magnetic field $\mathfrak{b}(\x)$. From the previous discussion, we infer that the integral kernel of the resolvent is not gauge invariant but it is gauge covariant: let $A':\R^2 \to \R^2$ be a magnetic potential that is gauge equivalent to $A$, namely such that there exists $\theta:\R^2 \to \R$ with $A-A'=\nabla \theta$, then
$$
U_\theta (H_b(A')-z\1)^{-1} U_\theta^{-1}= (H_b(A)-z\1)^{-1}
$$
which implies
$$
e^{\iu b \theta(\x)} \, (H_b(A')-z\1)^{-1}(\x;\y) \, e^{-\iu b \theta(\y)} = (H_b(A)-z\1)^{-1}(\x;\y) \, .
$$
The last equality can be rewritten by using the gradient theorem as 
\begin{equation}
\label{eq:gaugeCovRes}
e^{\iu b \Phi(\x,\y)}(H_b(A')-z\1)^{-1}(\x;\y)=(H_b(A)-z\1)^{-1}(\x;\y)
\end{equation}
where $\Phi(\x,\y)=\int_{\y}^{\x} \nabla \theta \cdot \mathrm{d}{\mathbf{s}}= \int_{\y}^{\x} A \cdot \mathrm{d}{\mathbf{s}} - \int_{\y}^{\x} A' \cdot \mathrm{d}{\mathbf{s}}$. This suggests to consider the integral kernel
$$
\mathcal{G}_b(\x;\y)=e^{-\iu b \int_\y^\x A \cdot \mathrm{d}{\mathbf{s}}} (H_b(A)-z\1)^{-1}(\x;\y) \, =  e^{-\iu b \varphi(\x,\y)} (H_b(A)-z\1)^{-1}(\x;\y) \, ,
$$
where we have used the definition of the Peierls phase $\varphi$, see Lemma \ref{lemma:MPhase}.
By using \eqref{eq:gaugeCovRes} one can see that, for any couple of equivalent gauges $A$ and $A'$,  we have 
$$
\mathcal{G}_b(\x;\y)=e^{-\iu b \int_\y^\x A \cdot \mathrm{d}{\mathbf{s}}} (H_b(A)-z\1)^{-1}(\x;\y) = e^{-\iu b \int_\y^\x A' \cdot \mathrm{d}{\mathbf{s}}} (H_b(A')-z\1)^{-1}(\x;\y) 
$$
which means that the operator $\mathcal{G}_b$ is actually a gauge invariant operator. 

If we write the equation \eqref{eq:ResolventID} as an equality between integral kernels (see also \eqref{eq:MPTResolvent}), and we multiply it on both sides by $e^{-\iu b \varphi(\x;\y)}$, we see that that \eqref{eq:ResolventID} is nothing but a convergent power series expansion in $\epsilon$ for the gauge invariant integral kernel of $\mathcal{G}_b$, where the zeroth element is nothing but the integral kernel of $(H_{b_0}-z\1)^{-1}$.

\section{Approximate perturbed projections}
\label{sec:Projections}
In this section we analyze the behaviour of spectral projections corresponding to bounded spectral islands of magnetic Schr\"odinger operators. The presentation is based on \cite{CorneanMonacoMoscolari2019, CorneanMonacoMoscolari2021}.

Let $\sigma_0$ be a spectral island of $H_{b_0}$ and let $\mathcal{C}\subset \rho(H_{b_0})$ be a positively oriented simple contour which encircles $\sigma_0$. From the results of the previous sections, see also Proposition \ref{Cornean2010}, we infer that $\mathcal{C}$ also belongs to $\rho(H_b)$ if $\epsilon=b-b_0$ is small enough, thus defining the bounded spectral island $\sigma_\epsilon$. Moreover, we can define two Riesz projections as
$$P_{b}= \frac{\iu}{2\pi} \oint_{\mathcal{C}} \, \left(H_b-z {\bf 1}\right)^{-1} \mathrm{d}z, \quad P_{b_0}= \frac{\iu}{2\pi} \oint_{\mathcal{C}} \, \left(H_{b_0}-z {\bf 1}\right)^{-1} \mathrm{d}z.
$$

Let us show that both $P_b$ and $P_{b_0}$ are integral operators with a jointly continuous integral kernel. We need a general technical lemma from \cite[Lemma B.7.8]{Simon1982}:
\begin{lemma}[\cite{Simon1982}]
\label{lemma:jointCont}
Let $C$ be a bounded operator on $L^2(\R^2)$ with an integral kernel  $C: \R^2 \times \R^2 \to \C$ with the property that the map 
$\R^2\ni \y \mapsto C(\cdot; \y)\in L^2(\R^2)$ is continuous in the $L^2$ norm. 

Let $A$ be an operator for which $A^*$ has an integral kernel that satisfies the same property as the integral kernel of $C$. Let $B$ be any bounded operator acting on $L^2(\R^2)$. Then there is a jointly continuous function $D(\x; \y)$ so that for $f, g \in L^{\infty}(\R^2)$ with compact support, we have
$$
\langle f, A B C g\rangle =\int_{\R^2\times \R^2} D(\x; \y) \overline{f(\x)} g(\y) \mathrm{d} \x \mathrm{d} \y .
$$
\end{lemma}

\begin{proof}
Let us define the $L^2$ functions $C_\y:=C(\cdot; \y) $ and $A^*_\x:=\overline{A(\x;\cdot)}$. Then it is enough to set 
$
D(\x;\y) := \langle A^*_\x , B C_\y \rangle \, . 
$
\end{proof}

Now we may write $P_b=ABC$ where we set $A=C=(H_b-\iu \mathbf{1})^{-1}$ and $B=(H_b-\iu  \mathbf{1})^2 P_b $. Since the integral kernel of $\left(H_b-\iu {\bf 1}\right)^{-1}$ only has a logarithmic singularity, see \eqref{IntKernelResolvent}, 
we have that 

$$\lim_{\x\to\x_0}\int_{\R^2} \Big |\big (H_b-\iu {\bf 1}\big)^{-1}(\x;\y)-\big (H_b-\iu {\bf 1}\big)^{-1}(\x_0;\y)\Big |^2\, \mathrm{d}\y=0,\quad \forall \; \x_0 \in \R^2,$$
hence, the hypotheses of Lemma \ref{lemma:jointCont} are satisfied and the integral kernel of $P_b$ is  jointly continuous. Moreover, by using the first resolvent identity one can write
$$P_b=\frac{\iu}{2\pi} \oint_{\mathcal{C}} \, (z-\iu)\big (H_b-z {\bf 1}\big )^{-1} \big (H_b-\iu  {\bf 1}\big )^{-1}\mathrm{d}z$$
and a simple application of Cauchy-Schwarz and triangle inequalities shows that there exist $C,\alpha>0$ such that
\begin{equation}
\label{eq:PExpo}
|P_b(\x;\y) | \leq C e^{-\alpha \|\x-\y\|} \, ,\quad \forall \, \x,\y\in \R^2.
\end{equation}

By mimicking the definition of $S_\epsilon(z)$, let us define the operator $\widetilde{P}_b$ via the integral kernel
\begin{equation}
\label{eq:TildeP}
\widetilde{P}_b(\x;\y):=e^{\iu \epsilon \varphi (\x,\y)} P_{b_0}(\x;\y)\, .
\end{equation}
The following proposition summarizes the main properties of $\widetilde{P}_b$:
\begin{proposition}
\label{prop:ApproxP'}
The operator $\widetilde{P}_b$ defined by the integral kernel in \eqref{eq:TildeP} is bounded and selfadjoint. Moreover, there exist $\epsilon_0,C,\alpha>0$ such that for all $0\leq |\epsilon|\leq \epsilon_0$ we have 
\begin{equation}
\label{eq:ApproxSeriesP}
P_b=\widetilde{P}_b+\frac{\iu }{2\pi}\sum_{n=1}^{\infty} \oint_{\mathcal{C}} (-1)^n\, S_\epsilon(z) (T_\epsilon(z))^n  \,\mathrm{d}z \, ,
\end{equation}
and 
\begin{align}
\label{eq:ApproxP}
&|P_b(\x;\y)-\widetilde{P}_b(\x;\y)| \leq |\epsilon|\,  C e^{-\alpha \|\x-\y\|} \, ,\quad \|P_b -\widetilde{P}_b\| \leq |\epsilon|\,  C \, .
\end{align}
\end{proposition}
\begin{proof}
The proof is a consequence of the Riesz integral expression of the Fermi projections together with the resolvent identity developed in the previous section, namely \eqref{eq:MPTResolvent} and \eqref{eq:ResolventID}.
\end{proof}

Proposition \ref{prop:ApproxP'} provides an approximation of the spectral projection $P_b$ in terms of bounded operators that have an explicit dependence on the magnetic field strength $b$, that is \eqref{eq:ApproxSeriesP}. 

Even though the operator $\widetilde{P}_b$ provides an explicit approximation for the projection $P_b$, it is often desirable to approximate the perturbed spectral projection $P_b$ with another orthogonal projection whose $b$ dependence can be explicitly controlled. For this purpose, define the operator  
\begin{equation*}
\Delta^{(b)}:=\big(\widetilde{\F}_{b}\big)^2-\widetilde{\F}_{b} \, .
\end{equation*}

Since $\widetilde{\F}_{b}$ is already a selfadjoint operator, the operator $\Delta^{(b)}$ measures how far $\widetilde{\F}_{b}$ is from being an orthogonal projection. By using \eqref{eq:CompositionId}, the identity $P_{b_0}^2=P_{b_0}$,   and a Taylor expansion, we get:
\begin{equation*} 
\begin{aligned}
&\Delta^{(b)}(\x;\x')\\
&=\int_{\R^2}  \left( e^{\iu \epsilon \varphi(\x,\y)} \F_{b_0}(\x;\y) e^{\iu \epsilon \varphi(\y,\x')} \F_{b_0}(\y;\x') - e^{\iu \epsilon \varphi(\x,\x')}  \F_{b_0}(\x;\y)  \F_{b_0}(\y;\x')  \right) \mathrm{d} \y  \\
&=\iu \, \epsilon \,  e^{\iu \epsilon \varphi(\x,\x')} \int_{\R^2}  \mathrm{fl} (\x,\y,\x')  \F_{b_0}(\x;\y) \F_{b_0}(\y;\x')\, \mathrm{d} \y\, + \mathcal{O}(\epsilon^2 e^{-\alpha \Vert \x-\x'\Vert})  \, .
\end{aligned}
\end{equation*}
From the exponential localization of the integral kernel of $P_{b_0}$, together with the second estimate of \eqref{eq:CompositionId}, we obtain
\begin{equation*}
|\Delta^{(b)}(\x;\x')| \leq  |\epsilon|\,  C e^{-\alpha\|\x-\x'\|} \, .
\end{equation*}
Via a Schur test \cite[Theorem 5.2]{HalmosSunder}, we conclude that $\|\Delta^{(b)}\| \leq |\epsilon|\, C$. The operator $\Delta^{(b)}$ is a function of the selfadjoint operator $\widetilde{\F}_{b}$ with a norm of order $|\epsilon|$. By functional calculus, the spectrum of $\widetilde{\F}_{b}$ needs to be concentrated near $0$ or $1$, which is compatible with the fact that $\widetilde{\F}_{b}$ is an almost projection. 
\begin{proposition}
\label{prop:ApproxP}
There exist $\epsilon_0>0$ and some constants $C,\alpha'>0$ such that for every $0\leq |\epsilon|\leq \epsilon_0$ we have that $\|\Delta^{(b)}\|<1/4$ and the operator 
    \begin{align}\label{ProjEps1}
\mathcal{P}^{(b)}:=\widetilde{\F}_{b}+\big (\widetilde{\F}_{b}-\tfrac{1}{2} {\bf 1}\big )\big \{ ({\bf 1}+4\Delta^{(b)})^{-1/2}-{\bf 1}\big \} \, 
\end{align}
is an orthogonal projection with a jointly continuous integral kernel which obeys: 
\begin{align}\label{hcj2}
    \big | \mathcal{P}^{(b)}(\x;\y)- \widetilde{\F}_{b}(\x;\y)\big |\leq C\, |\epsilon|\, e^{-\alpha'\|\x-\y\|}.
\end{align}
Moreover, let $\chi_L$ be the characteristic function of the set $\Lambda_L:=[-L,L]^2$, then $\chi_L P_b$ and $\chi_L \mathcal{P}^{(b)}$ are trace class and 
\begin{equation}
\label{eq:EqIDS}
\lim_{L \to \infty} \left|\frac{\Tr\left(\chi_L P_b\right)}{4L^2} - \frac{\Tr(\chi_L \mathcal{P}^{(b)})}{4L^2} \right| = 0 \, .
\end{equation}
\end{proposition}
\vspace{0.2cm}
\begin{remark}
    Recall that, given a spectral projection $P$, its integrated density of states $\mathcal{I}(P)$ is defined by 
$$
\mathcal{I}(P):= \lim_{L \to \infty} \frac{\Tr(\chi_L P)}{4L^2} 
$$
whenever $\chi_L P$ is trace class for every $L \in \R$, and the limit exists. Thus $P_b$ admits an integrated density of states if and only if $\mathcal{P}^{(b)}$ does, and in that case the two quantities are equal. The density of states can also be seen as the trace per unit surface of the spectral projection $P$. 
\end{remark}

\begin{proof}
In order to get a projection out of $\widetilde{\F}_{b}$, we use again the functional calculus, by following a construction which can be traced back at least to Nenciu \cite{Nenciu1993} (see also  \cite{Nenciu2002} ). The idea is to construct a projection corresponding to the spectrum of $\widetilde{\F}_{b}$ which is located near $1$. By explicit computation, we have
$$
(\widetilde{\F}_{b}-z\mathbf{1})\big (\widetilde{\F}_{b}-(1-z)\mathbf{1}\big )=z(1-z)\left(\1+\frac{\Delta^{(b)}}{z(1-z)}\right)
$$
which implies that if $\|\Delta^{(b)}\| <|z(1-z)|$, the operator $\left(\1+\frac{\Delta^{(b)}}{z(1-z)}\right)$ is invertible and:
\begin{equation}
\label{eq:ResF}
\begin{aligned}
(\widetilde{\F}_{b}-z\mathbf{1})^{-1}&=\frac{\widetilde{\F}_{b}+(z-1)\mathbf{1}}{z(1-z)}\left(\1+\frac{\Delta^{(b)}}{z(1-z)}\right)^{-1}\\
&={\left(\widetilde{\F}_{b}+(z-1)\mathbf{1}\right)}\left(\Delta^{(b)}-z(z-1)\1\right)^{-1}.
\end{aligned}
\end{equation}
The previous computation shows that if $\epsilon$ is small enough 
then $$\mathcal{C}_1:=\{z\in \C:\, |z-1|=1/2\}\subset \rho(\widetilde{\F}_{b}),$$ 
thus we may define 
$$
\mathcal{P}^{(b)}=\frac{\iu}{2\pi} \oint_{\mathcal{C}_1}  (\widetilde{\F}_{b}-z\1)^{-1} \mathrm{d} z \, .
$$
When $\|\Delta^{(b)}\|<1/4$ we have the identity 
\begin{equation}
\label{eq:CD}
\begin{aligned}
\left(\Delta^{(b)}-z(z-1)\1\right) &=-\left( \frac{\1+ (\1+4\Delta^{(b)})^{1/2}}{2}-z \1\right)\left( \frac{\1- (\1+4\Delta^{(b)})^{1/2}}{2}-z \1\right)\\
& =: -\left(C^{(b)}-z\1\right)\left(D^{(b)}-z\1\right) \, .
\end{aligned}
\end{equation}
For $\epsilon$ close to zero, the norm of $\Delta^{(b)}$ is also close to zero, which implies, by functional calculus, that the spectrum of $C^{(b)}$ is concentrated near $1$ and the spectrum of $D^{(b)}$ is concentrated near $0$. By inserting \eqref{eq:CD}
in \eqref{eq:ResF} and using the second resolvent identity, we get that $\mathcal{P}^{(b)}$ can be expressed as
$$
\begin{aligned}
\mathcal{P}^{(b)}&=\frac{\iu}{2\pi} \oint_{\mathcal{C}_1} {\left(\widetilde{\F}_{b}+(z-1)\mathbf{1}\right)} (\1+4\Delta^{(b)})^{-\frac{1}{2}} \left(\left(C^{(b)}-z\1\right)^{-1} - \left(D^{(b)}-z\1\right)^{-1} \right)  \mathrm{d} z\\
&=\frac{\iu}{2\pi} \oint_{\mathcal{C}_1} {\left(\widetilde{\F}_{b}+(z-1)\mathbf{1}\right)} (\1+4\Delta^{(b)})^{-\frac{1}{2}}\left(C^{(b)}-z\1\right)^{-1}   \mathrm{d} z\\
&=\widetilde{\F}_{b}(\1+4\Delta^{(b)})^{-\frac{1}{2}} +(\1+4\Delta^{(b)})^{-\frac{1}{2}}\frac{\iu}{2\pi} \oint_{\mathcal{C}_1} (z-1)\, \left(C^{(b)}-z\1\right)^{-1}   \mathrm{d} z \, ,
\end{aligned}
$$
where in the second equality we have used the fact that the contour $\mathcal{C}_1$ does not contain any part of the spectrum of $D^{(b)}$ for $\epsilon$ small enough, and in the third equality that 
$$\frac{\iu}{2\pi} \oint_{\mathcal{C}_1} \, \left(C^{(b)}-z\1\right)^{-1}   \mathrm{d} z =\1 \, ,$$
because the spectrum of $C^{(b)}$ is fully included in the integration contour. From the same reason, by writing $z\, \left(C^{(b)}-z\1\right)^{-1} =-\1 +C^{(b)}\left(C^{(b)}-z\1\right)^{-1}$ we get 
$$\frac{\iu}{2\pi} \oint_{\mathcal{C}_1} z\, \left(C^{(b)}-z\1\right)^{-1}   \mathrm{d} z=C^{(b)}$$
which eventually leads to \eqref{ProjEps1}. 

Now we want to prove \eqref{hcj2}. In order to analyze the integral kernel of $\mathcal{P}^{(b)}$, we first show a technical lemma, which is a generalization of \cite[Lemma 8.5]{CorneanMonacoMoscolari2019}.

\begin{lemma}
\label{lemma:IntExpo}
Let $
F(x):=\sum_{n=1}^{+\infty} a_n x^n 
$ be a power series with convergence radius $r>0$. Let $D_\epsilon$ be a bounded operator with an integral kernel $D_\epsilon(\x;\y)$ for which there exist three constants $\epsilon_0,C_0,\alpha_0>0$ such that for all $0\leq |\epsilon|\leq \epsilon_0$ we have
\begin{equation}
\label{eq:IntExpoLemma}
\left|D_\epsilon(\x;\y)\right| \leq |\epsilon|\,  C_0 e^{-\alpha_0 \|\x-\y\|}\, \, \, \, \, \forall \, \x\, , \y \in \R^2 \, , \, \, \text{ and } \, \, \epsilon_0\, C_0 \int_{\R^2} e^{-\alpha_0\|\x\|/2}\, \mathrm{d} \x <r.
\end{equation}
Then $\|D_\epsilon\|<r$ and $F(D_\epsilon)$ has an integral kernel obeying 
$$
|F(D_\epsilon)(\x;\y)|  \leq |\epsilon|\,  \frac{C_0}{r} \, e^{-\alpha_0 \|\x-\y\|/2} \, \quad \,\, \,  \forall \, \x\, , \y \in \R^2 \, .
$$
\end{lemma}
\begin{proof}[Proof of Lemma \ref{lemma:IntExpo}] Due to \eqref{eq:IntExpoLemma} and the Schur test \cite[Theorem 5.2]{HalmosSunder}  we have $\|D_\epsilon\|<r$. 
The operator $D_\epsilon^{n}$ with $n\geq 2$ has an integral kernel given by:
	\[ D_\epsilon^n(\x;\x')= \int_{\R^2} \mathrm{d}\x_1  \cdots \int_{\R^2} \mathrm{d}\x_{n-1}   D_\epsilon(\x;\x_1) \cdots D_\epsilon(\x_{n-1};\x')\, . \]
	By using the triangle inequality together with \eqref{eq:IntExpoLemma} we have:
	\begin{equation*}
		e^{\alpha_0 \|\x-\x'\|/2} \left|D_\epsilon^n(\x;\x')\right| 
		\leq |\epsilon|\,  C_0\,   \, r^{n-1}.
	\end{equation*}
 Thus the integral kernel of $F(D_\epsilon)$ is defined by the uniformly convergent series 
 $$
 F(D_\epsilon)(\x;\y)=\sum_{n=1}^{+\infty} a_n D_\epsilon^n(\x;\y).
 $$ 
\end{proof}
Now let us get back to the proof of Proposition \ref{prop:ApproxP}. The function $F(x)=(1+4x)^{-1/2}-1$ is real analytic near zero, while the integral kernel of $\Delta^{(b)}$ is exponentially localized and of order $\epsilon$, thus if $\epsilon_0$ is taken small enough,  \eqref{eq:IntExpoLemma} is satisfied. Hence \eqref{hcj2} follows from \eqref{ProjEps1} and Lemma~\ref{lemma:IntExpo}. 

It remains to show the limit \eqref{eq:EqIDS}. First, 
by putting together \eqref{hcj2} and \eqref{eq:ApproxP} we obtain that for all $0\leq |\epsilon|\leq \epsilon_0$ we have
\begin{equation}
\label{GoodLoc}
\left| \left(\F_{b}- \mathcal{P}^{(b)}\right)(\x;\x')\right|= \left|\left(\F_{b}-\widetilde{P}_{b} + \widetilde{P}_{b}- \mathcal{P}^{(b)}\right)(\x;\x') \right| \leq |\epsilon|\, C \,  e^{-\alpha\|\x-\x'\|}\, .
\end{equation}

Due to \eqref{GoodLoc} we have that $\|\F_{b}-\mathcal{P}^{(b)}\|<1$ when $\epsilon$ is sufficiently small. Since the norm difference of the two projections, $\F_{b}$ and $\mathcal{P}^{(b)}$, is strictly less than one, we can construct the Kato-Nagy operator $\mathsf{U}_b$ \cite[Sec. 1.6.8]{Kato}(see \eqref{eqn:KatoNagy} for the explicit expression), which is unitary and intertwines the two projections, namely $\mathsf{U}_b\F_{b}=\mathcal{P}^{(b)}\mathsf{U}_b$. The next Lemma is adapted from \cite[Lemma 8.5]{CorneanMonacoMoscolari2019}.
\begin{lemma} \label{Nagy}
	Let $\mathsf{U}_b$ be the Kato-Nagy unitary intertwining $P_{b}$ and $\mathcal{P}^{(b)}$, given by 
 \begin{equation} \label{eqn:KatoNagy}
		\begin{aligned}
			\mathsf{U}_b&=\left[{\bf 1}-\left({\mathcal{P}}^{(b)}-P_{b}\right)^{2}\right]^{-1/2} \left({\mathcal{P}}^{(b)} P_{b}+ \big ({\bf 1} - {\mathcal{P}}^{(b)} \big ) \left({\bf 1} - P_{b} \right) \right).
		\end{aligned}
	\end{equation}
 Then there exist $\epsilon_0,C,\alpha>0$ such that for all $0\leq |\epsilon|\leq \epsilon_0$ we have
	\begin{equation}
 \label{eq:LocKatoNagy}
	\left| \left(\mathsf{U}_b-{\bf 1}\right)(\x;\x') \right| \leq  |\epsilon|\,  C \,  \e^{-\alpha \|\x-\x'\|}\,, \qquad \forall \, \x\, , \y \in \R^2 .
	\end{equation}
\end{lemma}
\begin{proof}[Proof of Lemma \ref{Nagy}]
The function $F(x)=(1-x^2)^{-1/2}-1$ is real analytic on $(-1,1)$ and due to \eqref{GoodLoc} and Lemma \ref{lemma:IntExpo} we conclude that if $\epsilon_0$ is small enough, then $F\big ({\mathcal{P}}^{(b)}-P_{b}\big )$ has an integral kernel of order $\epsilon$ with an off-diagonal exponential decay.  By using that $P_{b}^2 +(\1-P_{b})^2=\1$, we observe that $\mathsf{U}_b-\1$ can be expressed as a sum of products involving at least one factor like 
 $F\big ({\mathcal{P}}^{(b)}-P_{b}\big )$ or ${\mathcal{P}}^{(b)}-P_{b}$, and we are done with proving \eqref{eq:LocKatoNagy}.
 \end{proof}

Finally, we want to prove \eqref{eq:EqIDS}.  By applying Lemma \ref{MainLemma} (whose proof is postponed to the Appendix \ref{sec:ContinuityDecay}) and  Lemma \ref{Nagy} to the projections $P_b$ and $\mathcal{P}^{(b)}$, we have that $\chi_L P_b$ and $\chi_L \mathcal{P}^{(b)}$ are trace class operators and 
$$\Tr (\chi_L P_b)-\Tr (\chi_L\mathcal{P}^{(b)})=\mathcal{O}(L)\, ,$$
thus \eqref{eq:EqIDS} immediately follows. 
\end{proof}

\section{Continuity in the strong topology and stability of the discrete spectrum}\label{sec4}

As a direct application of the operator-norm convergent expansion in \eqref{eq:ApproxSeriesP} we can prove strong convergence of magnetic projections.

\begin{proposition}
\label{prop:StrongTC}
The projection $P_b$ converges to $P_{b_0}$ in the strong operator topology.
\end{proposition}

\begin{proof}
Since  $C^{\infty}_0(\R^2)$ is dense in $L^2(\R^2)$, it is enough to prove the convergence for every function in $C^{\infty}_0(\R^2)$. Let $\eta \in C^{\infty}_0(\R^2)$ and consider 
$
\left\| (P_b-P_{b_0})\eta \right\|=\big \| (P_b-\widetilde{P}_b+ \widetilde{P}_b - P_{b_0})\eta \big \|
$. From \eqref{eq:ApproxP} it is enough to show that 
$$
\lim_{\epsilon\to 0} \big \| (\widetilde{P}_b - P_{b_0})\eta \big \|=0 \, . 
$$
Using that $$\big | e^{\iu \epsilon  \varphi(\x,\y)}-1 \big |\leq |\epsilon| \, |\varphi(\x,\y)|\leq C\, |\epsilon| \, \|\x\|\, \|\y\|,$$
we have the following estimate:
\begin{equation*}
\big | (\widetilde{P}_b - P_{b_0})(\x;\y) \big | \leq \big | e^{\iu \epsilon  \varphi(\x,\y)}-1 \big |\,  |P_{b_0}(\x;\y)| \leq C\, |\epsilon| \,\|\x\|\, \|\y\|\, e^{-\alpha \|\x-\y\|}.
\end{equation*}
Since $\eta$ has compact support, this estimate allows us to write 
$$\big |\big ((\widetilde{P}_b - P_{b_0})\eta\big )(\x)\big |\leq C_\eta\, |\epsilon|\, e^{-\alpha \|\x\|/2},$$
which finishes the proof.
\end{proof}

Actually we can prove more. Let $W_\epsilon=H_{b_0+\epsilon}-H_{b_0}=\iu 2\epsilon A\cdot \nabla + \iu \epsilon \,  \textrm{div} A + \epsilon^2A^2$. Then, due to the exponential localization of the integral kernel of $(H_{b_0}-z\1)^{-1} $ (see \eqref{IntKernelResolvent} and \eqref{eq:estimateResolvent}) and for some fixed $\alpha>0$ we have the estimate 
\begin{equation}\label{hcj4}
\sup_{z\in \mathcal{C}}\left\|W_\epsilon (H_{b_0}-z\1)^{-1}e^{-\alpha \|\cdot\|}\right\| \leq C\, |\epsilon| 
\end{equation}
and 
\begin{equation}\label{hcj5}
\sup_{z\in \mathcal{C}}\left\| \big ((H_b-z\1)^{-1}-(H_{b_0}-z\1)^{-1}\big )e^{-\alpha \|\cdot\|}\right\| \leq C\, |\epsilon|\, . 
\end{equation}

This implies that the reduced resolvent $$R_b(\mu):=(\1-P_b)(H_{b}-\mu \1)^{-1} =\frac{\iu }{2\pi}\oint_{\mathcal{C}}\frac{1}{\mu-z} (H_b-z\1)^{-1}\, \mathrm{d}z $$
(where $\mu$ belongs to the interior of the contour $\mathcal{C}$) has the same property: 
\begin{equation}\label{hcj6}
\sup_{{\rm dist}(\mu,\mathcal{C})\geq \delta>0} \left\| \big (R_b(\mu)-R_{b_0}(\mu)\big )\,e^{-\alpha \|\cdot\|}\right\| \leq C\, |\epsilon|. 
\end{equation}
\subsection{Perturbation of discrete eigenvalues}
\label{sec:Eigenvalue}
Let us now consider the situation where the spectral island $\sigma_0 \equiv \sigma_{b_0}$, enclosed by $\mathcal{C}$, consists of an isolated discrete eigenvalue $e_0$ of $H_{b_0}$, with multiplicity $M\geq 1$. 

From a result by Kato \cite[VIII Theorem 2.6]{Kato} (see also \cite[Section $\S$ 6]{AvronHerbstSimon}) we have that $\sigma_{b_0}$ will evolve into a spectral island $\sigma_b$ that contains $M$ discrete eigenvalues $e^{(i)}_b$ (including multiplicities).  Let us show how this fact can be proved within magnetic perturbation theory. 

From the standard spectral theory of Schr\"odinger operators \cite{ReedSimonIV} we have that the  eigenfunctions of $H_{b_0}$ corresponding to discrete eigenvalues are exponentially localized, namely, for every $\psi \in L^2(\R^2)$ such that $H_{b_0}\psi = e_0 \psi $, there exists two constants $C , \alpha>0$ such that for almost every $\x \in \R^2$
\begin{equation*}
|\psi(\x)| \leq C e^{-\alpha \|\x \|} \, .
\end{equation*}

In this setting, the family of magnetic projections is not only continuous in the strong operator topology but it is also continuous in the uniform topology.

\begin{proposition}
\label{prop:NormC}
Let $P_{b_0}$ and $P_b$ be the orthogonal projections onto the spectral island $\sigma_{b_0}$ and $\sigma_b$ respectively. Assume that the $\sigma_{b_0}=\{e_0\}$ and the dimension of the eigenspace of $e_0$ is equal to $M< \infty$. Then
$$
\lim_{b \to b_0} \|P_b-P_{b_0}\| =0\, .
$$
\end{proposition}
\begin{proof}
 As in the proof of Proposition \ref{prop:StrongTC}, we write
$ P_b-P_{b_0} =P_b-\widetilde{P}_b+ \widetilde{P}_b - P_{b_0} $. Since we have $\lim_{b \to b_0}\| P_b-\widetilde{P}_b\|=0$, it is left to show that $\lim_{b \to b_0} \|\widetilde{P}_b - P_{b_0}\| =0$. First, notice that the integral kernel of $P_{b_0}$ can be written as 
\begin{equation}
\label{eq:AuxPsiE}
P_{b_0}(\x;\y)=\sum_{i=1}^{M} {\psi_{i}(\x)} \overline{\psi_{i}(\y)} \, ,\quad |P_{b_0}(\x;\y)|\leq C\, e^{-\alpha\|\x\|} e^{-\alpha\|\y\|},
\end{equation}
where $\{\psi_i\}_{i=1}^M$ forms a basis of the eigenspace associated with $e_0$.
Then, we have the estimate
\begin{equation*}
\begin{aligned}
\left|\left(\widetilde{P}_b -P_{b_0}\right)(\x;\y) \right| &\leq \left| e^{\iu \epsilon \varphi(\x,\y)}-1 \right| \left|P_{b_0}(\x;\y)\right| \\ 
&
\leq |\epsilon|\, C \|\x\| \|\y\| e^{-\alpha(\|\x\|+\|\y\|)} \\
&\leq|\epsilon|\, C e^{-\alpha(\|\x\|+\|\y\|)/2} ,
\end{aligned}
\end{equation*}
where we have used that $|\varphi(\x,\y)| \leq C\, \|\x\| \|\y\|$. Then a Schur test \cite[Theorem 5.2]{HalmosSunder} finishes the proof.
\end{proof}

Proposition \ref{prop:NormC} shows that the family of projection $\{P_b\}$ is continuous with respect to the parameter $b$ in the uniform topology. By using the same strategy, and exploiting the explicit dependence on $b$ of $P_b$ in \eqref{eq:ApproxSeriesP}, one can also show that $P_b$ is actually smooth with respect to the parameter $b$. We will only sketch the proof of its differentiability. Namely, from \eqref{eq:ApproxSeriesP} we infer that 
$$P_b=\widetilde{P}_b-\frac{\iu }{2\pi}\oint_{\mathcal{C}} \, S_\epsilon(z) T_\epsilon(z)  \,\mathrm{d}z \, +\mathcal{O}(\epsilon^2)$$
in the uniform topology. Let $\mathfrak{W}_\epsilon$ be the operator whose integral kernel is given by 
\begin{align*}
&\mathfrak{W}_\epsilon(\x;\x') =2 e^{\iu \epsilon \varphi(\x,\x')}\frac{\iu }{2\pi}\oint_{\mathcal{C}} \int_{\R^2} (H_{b_0}-z\1 )^{-1}(\x;\y)\\
& \phantom{\,\,\,\, \,\, \mathfrak{W}_\epsilon(\x;\x') =2 e^{\iu \epsilon \varphi(\x,\x')}\frac{\iu }{2\pi}}\, \times \, A_{\x'}(\y)\cdot (-\iu \nabla_\y -b_0{A}(\y))(H_{b_0}-z {\bf 1})^{-1}(\y;\x') \, \mathrm{d}\y \, \mathrm{d}z.
\end{align*}
From \eqref{eqn:S(z)}, \eqref{eqn:T(z)} and \eqref{dhc1}, together with \eqref{IntKernelResolvent} and \eqref{eq:estimateResolvent}, we see that 
$$P_b=\widetilde{P}_b +\epsilon\, \mathfrak{W}_\epsilon +\mathcal{O}(\epsilon^2)$$
in the norm topology. By integrating with respect to $z$, and using that 
$$z\, \mapsto \, R_{b_0}(z)= (H_{b_0}-z\1 )^{-1}-\frac{1}{e_0-z}\, P_{b_0}$$
is analytic inside the integration contour, we obtain that $\mathfrak{W}_\epsilon$ consists of two terms, each containing a $P_{b_0}$ and a reduced resolvent $R_{b_0}(e_0)$. This implies that $\mathfrak{W}_\epsilon(\x;\x')$ is jointly exponentially localized in $\x$ and $\x'$, hence we may expand the exponential factors $e^{\iu \epsilon \varphi(\x,\x')}$ in both $\widetilde{P}_b$ and $\mathfrak{W}_\epsilon$ and obtain (we denote by $\Pi_1$ the operator with integral kernel $\iu \varphi(\x,\x')P_{b_0}(\x;\x')$):
$$P_b =P_{b_0} +\epsilon (\Pi_1+\mathfrak{W}_0) +\mathcal{O}(\epsilon^2),$$
which shows that the map $b\mapsto P_b$ is differentiable in the norm topology and its derivative at $b_0$ equals $\Pi_1+\mathfrak{W}_0$.  

Furthermore, Proposition \ref{prop:NormC} implies that for $b-b_0$ small enough, $P_b$ and $P_{b_0}$ are intertwined by a Kato-Nagy unitary $U_b$ (see \eqref{KN}), namely we have $P_b U_b = U_b P_{b_0}$. Since $P_b$ admits an asymptotic power series in $b-b_0$, see \eqref{eq:ApproxSeriesP}, the same holds for $U_b$, and, by following the strategy to prove Lemma \ref{Nagy}, we get that there exist $\alpha,C>0$ such that 
\begin{equation}
\label{eq:KNDiscrete}
|(U_b-\1)(\x;\y)| \leq |\epsilon| C  e^{-\alpha \|\x-\y\|}\, .
\end{equation}
The spectral island $\sigma_b$ will consist of exactly $M$ discrete eigenvalues of the finite dimensional matrix given by 
$$P_{b_0}U_b^* H_b U_b P_{b_0},\quad T_{ij}(b):=\langle U_b\psi_i, H_b U_b\psi_j\rangle,\quad P_{b_0}\psi_i=\psi_i.$$

Let us for simplicity assume that $M=1$, i.e. $e_0$ is simple and $H_{b_0}\psi=e_0\psi$. We have that $U_b\psi$ is an eigenvector for $H_b$ for the eigenvalue $e_\epsilon$, also exponentially localized, and: 
$$P_b H_b=\frac{\iu}{2\pi} \oint_{\mathcal{C}} z(H_b-z\1)^{-1} \mathrm{d}z,\quad e_\epsilon=\frac{\iu}{2\pi} \oint_{\mathcal{C}} \langle U_b\psi, z(H_b-z\1)^{-1}U_b\psi\rangle \, \mathrm{d}z$$
which, together with \eqref{hcj5} and \eqref{eq:KNDiscrete}, shows that $|e_\epsilon-e_0|\leq C\, |\epsilon|$.  Moreover, by using that $W_\epsilon=H_{b}-H_{b_0}$ we have 
$$H_b(P_b-P_{b_0})=H_bP_b - (H_b-H_{b_0})P_{b_0}-H_{b_0}P_{b_0}=e_\epsilon P_b-e_0P_{b_0} -W_\epsilon P_{b_0}$$
hence 
\begin{equation}\label{hcj7}
\Vert H_b(P_b-P_{b_0})\Vert \leq C\, |\epsilon|. 
\end{equation}

We now follow the strategy from Savoie \cite{Savoie}, who used the Feshbach formula in order to get a much better location of the perturbed eigenvalue $e_\epsilon$, up to an error of order $\epsilon^3$. The next lemma is necessary for the control of the resolvent of $H_b$ restricted to $\1-P_{b_0}$:
\begin{lemma}\label{lemmahc4}
    Let $\widetilde{H}_b:=(\1-P_{b_0})H_b (\1-P_{b_0}).$ Then if $|\mu-e_0|\leq C |\epsilon|$, we have that $\widetilde{H}_b-\mu (\1-P_{b_0})$ is invertible on the range of the operator $\1-P_{b_0}$, and if $R_{b}(\mu)=(\1-P_{b})(H_{b}-\mu\mathbf{1})^{-1}$ is the reduced resolvent of $H_b$ on the range of the operator $\1-P_b$, then 
    \begin{equation}\label{hcj8}
 \left\| \big (\widetilde{H}_b-\mu (\1-P_{b_0})\big )^{-1}-(\1-P_{b_0})\, R_{b}(\mu)(\1-P_{b_0})\, \right\| \leq C\, |\epsilon|. 
\end{equation}
\end{lemma}
\begin{proof}
Using \eqref{hcj7} we get 
$$\widetilde{H}_b-\mu (\1-P_{b_0}) =(\1-P_b)(H_b-\mu \1)(\1-P_b)+ \mathcal{O}(\epsilon)$$
hence 
\begin{equation}
\label{eq:F}
\begin{aligned}
    &\Big (\widetilde{H}_b-\mu (\1-P_{b_0})\Big )\, (\1-P_{b_0})\, R_b(\mu)\, (\1-P_{b_0})\\
    &\qquad =(\1-P_{b_0})-(\1-P_{b_0})[W_\epsilon,P_{b_0}]R_b(\mu)(\1-P_{b_0})+\mathcal{O}(\epsilon)\\
    &\qquad =(\1-P_{b_0})\big (\1+\mathcal{O}(\epsilon)\big )(\1-P_{b_0}),
\end{aligned}
\end{equation}
where in the last equality we have used \eqref{hcj4} to estimate the norm of the commutator (which in principle can be first defined as a quadratic form on $C^{\infty}_0(\R^2)$). Therefore the r.h.s of \eqref{eq:F} is invertible when $|\epsilon|$ is small enough.
\end{proof}

The Feshbach formula \cite{Feshbach,Savoie} applied to $H_b$ and $P_{b_0}=|\psi\rangle\langle\psi|$ states that a real number $\mu$ which obeys $|\mu-e_0|\leq C|\epsilon|$ is an eigenvalue for $H_b$ if and only if 
\begin{equation}\label{hcj10}
\mu=e_0+\langle \psi,W_\epsilon \psi\rangle  -\langle W_\epsilon\psi, \big (\widetilde{H}_b-\mu (\1-P_{b_0})\big )^{-1}W_\epsilon\psi\rangle. 
\end{equation}
This is nothing but a fixed-point equation whose iterations starting from $\mu=e_0$ give us an asymptotic expansion in $\epsilon$ of $e_\epsilon$. We a-priori know that $e_\epsilon=e_0+\mathcal{O}(\epsilon)$ and $W_\epsilon\psi =\mathcal{O}(\epsilon)$, hence using \eqref{hcj8} and \eqref{hcj6}, together with the exponential localization of the eigenfunction (see \eqref{eq:AuxPsiE}) and \eqref{eq:estimateResolvent},
we obtain
$$\langle W_\epsilon\psi, \big (\widetilde{H}_b-\mu (\1-P_{b_0})\big )^{-1}W_\epsilon\psi\rangle-\langle W_\epsilon\psi, R_{b_0}(e_0)W_\epsilon\psi\rangle=\mathcal{O}(\epsilon^3),$$
which leads to 
$$e_\epsilon=e_0+ \langle \psi,W_\epsilon \psi\rangle  -\langle W_\epsilon\psi, R_{b_0}(e_0)W_\epsilon\psi\rangle\, + \mathcal{O}(\epsilon^3).$$

\section{Lack of continuity in the norm topology}\label{sec5}
Proposition \ref{prop:ApproxP}, in particular \eqref{eq:EqIDS}, together with Lemma \ref{MainLemma}, plays a crucial role in \cite{CorneanMonacoMoscolari2021}. By exploiting the gauge covariant magnetic perturbation theory described here, in \cite{CorneanMonacoMoscolari2021} the magnetic derivative of the integrated density of states has been analyzed in the context of Bloch-Landau Hamiltonians, namely under the assumptions that the scalar potential $V$ is smooth and $\Z^2$ periodic, and the magnetic field can be decomposed into a constant magnetic field plus a periodic field with zero-mean flux, that is the total magnetic vector potential can be written as $bA_L+\mathcal{A}$ for $b \in \R$ and $\mathcal{A}$ is smooth and $\Z^2$ periodic. In such a framework, it has been showed \cite{CorneanMonacoMoscolari2021} that the integrated density of states of $P_b$ exists, and it is a linear function of the magnetic field. In particular, we have that
\begin{equation}
\label{eq:JEMS}
\mathcal{I}(P_b)= c_0 + \frac{b}{2\pi} c_1
\end{equation}
where $c_0\in \mathbb{Q}$ and $ c_1\in \Z$ are constants that characterize the spectral gaps surrounding the spectral island $\sigma_{b_0}$, which means that \eqref{eq:JEMS} is valid as long as the magnetic perturbation does not close the spectral gaps defining the isolated part of the spectrum $\sigma_b$. The constants $c_0$ and $c_1$ can be explicitly computed: in particular,  the constant $c_1$ coincides with the Chern character of the projection $P_b$ (and of $P_{b_0}$) defined in \eqref{eq:Chern}, see \cite{CorneanMonacoMoscolari2021} for more details.

After having established the convergence in the strong operator topology, it is natural to analyze the family of projections in the uniform topology. However, magnetic projections are generally not continuous in $b$ with respect to the uniform topology: 
\begin{proposition}\label{prophc5}
If the Chern character of $P_{b_0}$, $C(P_{b_0}),$ is different from zero, namely, if in \eqref{eq:JEMS} $c_1\neq 0$, then it must hold that $\|P_b-P_{b_0}\|=1$ for all $b\neq b_0$ in a small interval around $b_0$.
\end{proposition}
\begin{proof}
    Assume by contradiction that for some $|b-b_0|>0$ small enough, we have $\|P_b-P_{b_0}\| <1$. Then, there exists a Kato-Nagy unitary $U_b$ that intertwines between $P_b$ and $P_{b_0}$. From Lemma \ref{MainLemma'} we have that  $U_b-\1$ has an exponentially localized integral kernel. By applying Lemma \ref{MainLemma}, we have that $\mathcal{I}(P_b)=\mathcal{I}(P_{b_0})$, but this is in contradiction with \eqref{eq:JEMS}. Hence $\|P_b-P_{b_0}\|=1$ (remember that the norm of the difference of two projections is always less or equal than one).
\end{proof}
The above failure of the norm continuity relies on the fact that the Chern character of the spectral projection is different from zero. To show a simple explicit example, let us consider the situation of the pure Landau Hamiltonian $H_b=(-\iu \nabla -b A_L)^2$ with $b>0$. The lowest spectral island is $\sigma_b=\{b\}$, an eigenvalue of infinite multiplicity, whose corresponding spectral projection $P_b$ is given by the integral kernel
$$P_b(\x;\y)=\frac{b}{2\pi} e^{\iu b \varphi(\x,\y)}\, e^{-b\|\x-\y\|^2/4},$$
where $\varphi(\x,\y)$ is given in \eqref{hcj1}. Then $\mathcal{I}(P_b)=b/(2\pi)$, namely  $c_0=0$ and $c_1=1$ in \eqref{eq:JEMS}. This shows that $P_b$ is nowhere continuous in the uniform norm. 

Finally, it is worth mentioning that magnetic spectral projections might fail to be norm continuous even in the topologically trivial case, i.e. when their Chern character is equal to zero, as it is shown in \cite[Corollary 1.2]{CorneanMonacoMoscolari2021} (see also the discussions in \cite{Nenciu1991, CorneanHerbstNenciu}) by using the concept of exponentially localized magnetic Wannier functions.

\appendix

\section{Continuity and off-diagonal decay}
\label{sec:ContinuityDecay}
In this section we collect three more general technical lemmas that we used in the previous proofs. Lemma \ref{lemmahc1} and \ref{MainLemma'} deal with the continuity and exponential localization of the integral kernel of an operator power series. They are a generalization of the result contained in \cite[Appendix C]{CorneanMonacoMoscolari2021}. Moreover, Lemma \ref{MainLemma} provides an estimate on the difference of local traces for projections that can be intertwined by a Kato-Nagy unitary. 

We start with a generalization of the result of Lemma \ref{lemma:IntExpo}. Instead of requiring a uniform control on the $\epsilon$ dependence of the integral kernel of the operator $D_\epsilon$, we only require the exponential localization of the integral kernel and the control on the smallness of the operator norm.

\begin{lemma}\label{lemmahc1}
Let $F(x)=\sum_{n\geq 1} a_nx^n$ be a power series with convergence radius $r>0$. Let $D$ be a bounded operator with a jointly continuous integral kernel, for which there exist $C,\alpha>0$ such that 
$$|D(\x;\y)|\leq C\, e^{-\alpha\|\x-\y\|}\, \qquad \forall \, \x\, , \y \in \R^2 .$$
If $\|D\|<r$, then there exist $0<\alpha'<\alpha$ and $C'>0$ such that $F(D)$ also has a jointly continuous integral kernel and 
$$\left|\big (F(D)\big )(\x;\y)\right|\leq C'\, e^{-\alpha'\|\x-\y\|}\, \quad \forall \, \x\, , \y \in \R^2 .$$
\end{lemma}
\begin{proof}
We first show that there exists $C>0$ such that for all $\x_0\in \R^2$ and $0<\alpha'<\alpha/2$ we have  
\begin{equation}\label{hcj3}
\|e^{\alpha' \| \cdot -\x_0\|} De^{-\alpha' \|\cdot -\x_0\|}-D\| \leq C \alpha'.
\end{equation}
Indeed, the integral kernel of that operator is 
$$\Big (e^{\alpha' \big (\| \x -\x_0\|- \|\y -\x_0\|\big )}-1\Big ) D(\x;\y).$$
From the reverse triangle inequality we have $| \|\x-\x_0\|-\|\y-\x_0\| | \leq \| \x-\y\|$, which coupled with the inequality $|e^{|x|}-1|\leq |x| e^{|x|}$ gives
$$
\Big | e^{\alpha' \big (\| \x -\x_0\|- \|\y -\x_0\|\big )}-1\Big |\leq \Big | e^{\alpha' \big |\| \x -\x_0\|- \|\y -\x_0\|\big |}-1\Big |\leq \alpha'\|\x-\y\| e^{\alpha' \|\x-\y\|} \, .
$$
Thus for all $\x,\y,\x_0\in \R^2$ we have 
$$\Big | \Big (e^{\alpha' \big (\| \x -\x_0\|- \|\y -\x_0\|\big )}-1\Big ) D(\x;\y)\Big |\leq \alpha' \, C\, e^{-\alpha \| \x -\y\|/4},$$
and the Schur test \cite[Theorem 5.2]{HalmosSunder} finishes the proof of \eqref{hcj3}. Thus, if $\alpha'$ is small enough then $$\|e^{\alpha' \| \cdot -\x_0\|} De^{-\alpha' \|\cdot -\x_0\|}\|<r.$$ 

We have
$$F(x)=a_1x +x^2 G(x),\quad G(x)=\sum_{n\geq 0}a_{n+2}\, x^n,$$
where $G$ has the same convergence radius. Then 
$$F(D)=a_1\, D + D\, G(D)\, D, $$
and 
$$e^{\alpha' \| \cdot -\x_0\|} G(D)e^{-\alpha' \|\cdot -\x_0\|}=G\big (e^{\alpha' \| \cdot -\x_0\|} De^{-\alpha' \|\cdot -\x_0\|}\big )$$
are bounded operators. Moreover, $DG(D)D$ is an integral operator with an integral kernel given by 
$$ \big (DG(D)D\big )(\x;\y)= \langle \overline{D(\x;\cdot)}, G(D) D(\cdot; \y)\rangle\, . $$
Let us fix a compact set $\Omega\subset \R^2$. Let $n\geq 1$. Due to the exponential localization of $D$, there exists another compact set $\Omega_n$ containing $\Omega$ such that 
$$\sup_{\x,\y\in \Omega}\Big |\langle \overline{D(\x;\cdot)}, G(D) D(\cdot; \y)\rangle -\langle \overline{D(\x;\cdot)}, \chi_{\Omega_n}G(D) \chi_{\Omega_n}D(\cdot; \y)\rangle\Big |\leq 1/n.$$
We apply Lemma \ref{lemma:jointCont} to $D\chi_{\Omega_n}G(D) \chi_{\Omega_n}D$ and we get that the scalar products
$$\langle \overline{D(\x;\cdot)}, \chi_{\Omega_n}G(D) \chi_{\Omega_n}D(\cdot; \y)\rangle$$ define jointly continuous functions, hence $\langle \overline{D(\x;\cdot)}, G(D) D(\cdot; \y)\rangle$ is also jointly continuous as a uniform limit. 

Finally, putting $\x_0=\y$ we have 
$$
\begin{aligned}
&\langle \overline{D(\x;\cdot)}, G(D) D(\cdot;\y)\rangle=\langle \overline{D(\x;\cdot)}e^{-\alpha' \| \cdot -\y\|}, G\big (e^{\alpha' \| \cdot -\y\|} De^{-\alpha' \|\cdot -\y\|}\big ) e^{\alpha' \|\cdot -\y\|} D(\cdot; \y)\rangle
\end{aligned}
$$
which shows that it is uniformly bounded by $C'\, e^{-\alpha'\|\x-\y\|}$. 
\end{proof}

The next result generalizes Lemma \ref{Nagy} to the setting in which one can control only the norm difference of the projections instead of the difference of the integral kernels. The proof is based on Lemma \ref{lemmahc1}.

\begin{lemma}
	\label{MainLemma'}
	Let $P_1$ and $P_2$ be two orthogonal projections which have jointly continuous integral kernels such that 
 $$|P_j(\x,\y)|\leq C\, e^{-\alpha\|\x-\y\|} \qquad \forall \, \x\, , \y \in \R^2 ,\, j\in \{1,2\},$$
 for some fixed $\alpha>0$ and $C<\infty$, and $\|P_1-P_2\|<1$.    
Then the Kato-Nagy intertwining unitary $U$ given by 
\begin{equation} \label{KN}
		\begin{aligned}
			U&=\left[{\bf 1}-\left(P_1-P_2\right)^{2}\right]^{-1/2} \big (P_1P_2+ \left({\bf 1} - P_1 \right) \left({\bf 1} - P_{2} \right) \big )
		\end{aligned}
	\end{equation}
and for which $P_1U=U P_2$, has the property that $U-\1$ has a jointly continuous integral kernel and there exist $C'>0$ and $0<\alpha'<\alpha$ such that 
$$|(U-\1)(\x;\y)|\leq C'\, e^{-\alpha'\|\x-\y\|}.$$
  \end{lemma}
  \begin{proof}
 The function $F(x)=(1-x^2)^{-1/2}-1$ defines an absolutely convergent power series as required by Lemma \ref{lemmahc1}, with convergence radius $1$. The only "non-trivial" term in $U-\1$ is the one containing the factor $F(P_1-P_2)$, which according to Lemma~\ref{lemmahc1}, it has a jointly continuous and exponentially decaying integral kernel.  
  \end{proof}

\begin{lemma}
	\label{MainLemma}{\cite[Lemma 2.1]{CorneanMonacoMoscolari2021}}
	Let $P_1$, $P_2$ and $U$ as in  Lemma \ref{MainLemma'}. For $L\geq 1$, let $\chi_L$ be the characteristic function of the set $\Lambda_L:=[-L,L]^2$. Then $\chi_L P_j$, $j \in \{1,2\}$, are trace class, and there exists another constant $C_2<\infty $ such that 
	\[
	\left|  \Tr(\chi_L P_1)-\Tr(\chi_L P_2)\right| \leq  L C_2\, .
	\]
	
\end{lemma}

\begin{proof} We may write 
$$\chi_L P_j= \big (\chi_L P_j e^{ \alpha'\|\cdot\|/2}\big ) \, \big (e^{-{ \alpha'}  \|\cdot\|/2}P_j\big )$$
hence $\chi_LP_j$ is a product of two Hilbert-Schmidt operators, thus trace class. Moreover, we have 
$$\chi_LUP_j=\chi_LP_j +\big (\chi_L(U-\1)e^{ \alpha' \|\cdot\|/2}\big ) \, \big (e^{-\alpha'  \|\cdot\|/2}P_j\big ) \, ,$$
thus $\chi_LUP_j$ are also trace class. We also have the identity 
$$\chi_LP_1-\chi_LP_2=\chi_L U P_2 U^*-\chi_L P_2 U^*U \, .$$
By exploiting the invariance under unitary conjugation of the trace, we obtain the equality \mbox{$\Tr (\chi_L P_2 U^*U)=\Tr (U\chi_L P_2 U^*)$}, thus
	$$\Tr(\chi_L P_1)-\Tr(\chi_L P_2)=\Tr \left ( [\chi_L,U]P_2U^*\right ).$$
	Denoting by $W:=U-\mathbf{1}$ we see that $W(\x;\x')$ is exponentially localized near the diagonal and
	$$\Tr \left ( [\chi_L,U]P_2U^*\right )=\Tr \left ( [\chi_L,W]P_2\right )+\Tr \left ( [\chi_L,W]P_2 W^*\right ).
    $$
	Both traces on the r.h.s. can be computed as an integral over $\R^2$ of the diagonal value of the corresponding integral kernels. Furthermore, both traces can be bounded by a double  integral of the type 
	\begin{equation}
 \label{eq:DoubleInt}
 \int_{\x\in \R^2}\int_{{\x'}\in \R^2}e^{-{\alpha'} \Vert \x-\x'\Vert}|\chi_L(\x)-\chi_L(\x')|\mathrm{d}\x' \mathrm{d}\x\; .
 \end{equation}
	In the integral \eqref{eq:DoubleInt}, the integrand is non-zero only if one variable belongs to $\Lambda_L$ and the other one lies outside $\Lambda_L$. Then, due to symmetry, it is enough to estimate the integral
	$$\int_{\x\in \Lambda_L}\int_{{\x'}\in \R^2\setminus \Lambda_L}e^{-\alpha'  \Vert \x-\x'\Vert}\mathrm{d}\x' \mathrm{d}\x\; .$$
	Notice that for a fixed $\x\in \Lambda_L$ we have the inequality
	$$e^{- \alpha'  \Vert \x-\x'\Vert}\leq e^{-\alpha' \;  {\rm dist}(\x,\partial \Lambda_L)/2}e^{- \alpha'  \Vert \x-\x'\Vert /2},\quad \forall \, \x'\in \R^2\setminus \Lambda_L,$$
    where we have denoted by $\partial \Lambda_L$ the boundary of the set $\Lambda_L$.
	Finally, by integrating with respect to $\x'$ at fixed $\x$ we can bound the above double integral by
	$$\int_{\x\in \Lambda_L}e^{- \alpha' \;  {\rm dist}(\x,\partial \Lambda_L)/2}\mathrm{d}\x\leq C L \,,$$
	thus finishing the proof.  
\end{proof}

\section{Integral kernel estimates for the resolvent}
\label{sec:Estimates}
In this section we briefly show how to prove the estimates \eqref{IntKernelResolvent} and \eqref{eq:estimateResolvent}.  Let us recall our assumptions on the scalar and magnetic potentials: we assume that the scalar potential $V$ is in $L^{2+\delta}_{\mathrm{u.loc}}(\R^2)$, for $\delta>0$, which means that it is uniformly in $L^{2+\delta}_{\rm loc}(\R^2)$, namely
\begin{equation}
\label{eq:uloc}
\|V\|_{2+\delta, \textrm{unif}}:=\sup_{\y \in \R^2} \int_{\|\x-\y\|\leq 1}  |V(\x)|^{2+\delta} \mathrm{d}\x < +\infty \, ,
\end{equation}
and $A,\mathcal{A}\in C^{\infty}(\R^2,\R^2)$. 

With this assumption, it is a standard result that $V$ is infinitesimally operator bounded with respect to the free Laplacian \cite{ReedSimonIV}. Moreover, since the magnetic potential is smooth, we have that $H_b$ satisfies the diamagnetic inequality \cite{AvronHerbstSimon}, this means that $V$ is also infinitesimally operator bounded with respect to the magnetic Laplacian $(-\iu \nabla -bA-\mathcal{A})^2$. From  \cite{BroderixHundertmarkLeschke,Simon1982} (see also \cite[Proposition 5.5]{Nenciu2002}) we get that the resolvent of $H_b$ has an integral kernel that is jointly continuous outside of the diagonal. Without entering into the details of the argument, the main idea is to analyze the Schr\"odinger semigroup, namely the operator $e^{-tH_b}$, for $t>0$, and then, for $\lambda>0$ large enough, to use the Laplace transform:
$$
(H_b+\lambda{\bf{1}})^{-1}=\int_{0}^{+\infty}  e^{-t H_b} e^{-t \lambda} \, \mathrm{d}t\, .
$$
 The key observation lies in the fact that the diamagnetic inquality allows a pointwise upper bound of the Schr\"odinger semigroup of $H_b$ provided by the Schr\"odinger semigroup of the free Laplacian, which is explicitly known. This leads to \eqref{IntKernelResolvent} with $z$ replaced by $-\lambda$. The estimate for an arbitrary $z$ is then obtained by repeatedly using the first resolvent identity. 

It remains to show \eqref{eq:estimateResolvent}, which is a bit more delicate. First, notice that for smooth and bounded magnetic fields, the same estimate has been proved in \cite[Proposition B.9]{CorneanFournaisFrankHelffer}, in particular for the purely magnetic operator $h_b:=(-\iu \nabla -b A -\mathcal{A})^2$. From the second resolvent identity we get
$$
(H_b-z{\bf{1}})^{-1}=(h_b-z{\bf{1}})^{-1} + (h_b-z{\bf{1}})^{-1} V (H_b-z{\bf{1}})^{-1} \, ,
$$
thus we are left to show that the integral kernel
$$
\Big (\big (-\iu \nabla_\x -bA-\mathcal{A})\, (h_b-z{\bf{1}})^{-1}\big ) \, V \, (H_b-z{\bf{1}})^{-1}\Big )(\x;\x') 
$$
satisfies \eqref{eq:estimateResolvent}. In view of \eqref{IntKernelResolvent}, it is enough to show that the estimate \eqref{eq:estimateResolvent} is satisfied by the following integral 
$$
\mathfrak{I}(\x,\x'):=\int_{\R^2}  e^{-\alpha\|\x-\y\|}\, \|\x-\y\|^{-1} |V(\y)| e^{-\alpha\|\y-\x'\|}\, \ln (2+\|\y-\x'\|^{-1}) \mathrm{d} \y\, . 
$$
In order to exploit the fact that $V$ is uniformly locally in $L^{2+\delta}(\R^2)$, we start by bounding the integral $\mathfrak{I}(\x,\x')$ by
$$
\mathfrak{I}(\x,\x')\leq \sum_{\gamma \in \Z^2} \mathfrak{I}_\gamma(\x,\x')
$$
where
$$\mathfrak{I}_\gamma(\x,\x'):= \int_{\|\y-\gamma \| \leq 1} e^{-\alpha\|\x-\y\|}\, \|\x-\y\|^{-1} |V(\y)| e^{-\alpha|\y-\x'|}\, \ln (2+\|\y-\x'\|^{-1})  \mathrm{d} \y \, .  $$
Then, notice that for $\y$ in the unit ball centered in $\gamma\in \Z^2$  we have 
$$\|\x-\gamma\|\leq \|\x-\y\|+ 1 \, ,$$ 
which, together with the triangle inequality $ \|\x-\x'\|\leq \|\x-\y\|+\|\y-\x'\|$ implies
\begin{equation}\label{eq:integrabilityG}
\begin{aligned}
\left|\mathfrak{I}_\gamma(\x,\x')\right|&\leq C\, e^{-\alpha\|\x-\x'\|/4} e^{-\alpha \|\x-\gamma\|/4} \\
&\qquad \times \int_{\|\y-\gamma\|\leq 1}\|\x-\y\|^{-1} |V(\y)| \, \ln (2+\|\y-\x'\|^{-1})\, \mathrm{d}\y.
\end{aligned}
\end{equation}
The estimate \eqref{eq:integrabilityG}  guarantees the summability in $\gamma$, and the desired exponential decay in $\|\x-\x'\|$. Therefore it remains to show that the remaining integral in \eqref{eq:integrabilityG} cannot diverge worse than $\|\x-\x'\|^{-1}$. 
For this purpose, we may assume that $0<\|\x-\x'\|\leq 1$. Let us then focus on 
$$\int_{\|\y-\gamma\|\leq 1}\|\x-\y\|^{-1} |V(\y)| \, \ln (2+\|\y-\x'\|^{-1})\, \mathrm{d}\y.$$
On the region $\|\y-\x'\|\leq \|\x-\x'\|/2$ we can bound the logarithmic term with $\ln(2+2\|\x-\x'\|^{-1})$ and consider 
\begin{equation}
\label{eq:AuxInt}
\int_{\|\y-\gamma\|\leq 1}\|\x-\y\|^{-1} |V(\y)|\, \mathrm{d}\y.
\end{equation}
The integral \eqref{eq:AuxInt} can then be bounded by using H\"older's inequality with $p=2+\delta $ and $q=(2+\delta)/(1+\delta)\in (1,2)$
and the fact that that $\|\x-\cdot\|^{-1}$ belongs to $L^q_{\rm loc}(\R^2)$. 

On the region $\|\x-\y\|\leq \|\x-\x'\|/2$ we may bound $\|\x-\y\|^{-1}$ by $2\|\x-\x'\|^{-1}$, while the remaining integral can be bounded in the same way as \eqref{eq:AuxInt}, because the logarithm also belongs to $L^q_{\rm loc}(\R^2)$. 

Finally, on the region where both $\|\x-\x'\|/2\leq \|\x-\y\|$ and 
$\|\x-\x'\|/2\leq \|\x'-\y\|$ we can bound the integral by 
$$\ln(2+2\|\x-\x'\|^{-1})\int_{\|\x-\x'|/2\leq \|\y-\x\|\leq 2}\|\x-\y\|^{-1} |V(\y)|\, \mathrm{d}\y ,$$
which, by using again H\"older's inequality with $p=2+\delta $ and $q=(2+\delta)/(1+\delta)\in (1,2)$, can be bounded by 
$$\ln(2+2\|\x-\x'\|^{-1})\, \|V\|_{p,{\rm unif}}\, \Big ( \int_{\|\x-\x'\|/2}^2 r^{-q+1} \mathrm{d} r \Big )^{1/q} \leq C\, \ln(2+2\|\x-\x'\|^{-1}).$$

\subsection*{Acknowledgments}
H.D.C. gratefully acknowledges the financial
support from Grant 2032-00005B of the Danish Council for Independent Research $|$ Natural Sciences. 
M.M. gratefully acknowledges the support of PNRR Italia Domani and Next Generation EU through the ICSC National Research Centre for High Performance Computing, Big Data and Quantum Computing and the support of the MUR grant Dipartimento di Eccellenza 2023–2027.

\phantom{authors}
\vfill
\hspace{-0.3cm}
{\begin{tabular}{rl}
	    (H.~D. Cornean) & \textsc{Department of Mathematical Sciences, Aalborg University} \\
		&  Thomas Manns Vej 23, 9220 Aalborg, Denmark \\
		&  \textsl{E-mail address}: \href{mailto:cornean@math.aau.dk}{\texttt{cornean@math.aau.dk}} \\
		\\
		(M. Moscolari) & \textsc{Dipartimento di Matematica, Politecnico di Milano}\\
		& Piazza Leonardo da Vinci 32, 20133 Milano, Italy \\
		&  \textsl{E-mail address}: \href{mailto:massimo.moscolari@polimi.it}{\texttt{massimo.moscolari@polimi.it}} \\
	\end{tabular}
	
}

\end{document}